\documentclass[journal]{IEEEtran}
\ifCLASSINFOpdf
\else
\fi
\newcommand{\eat}[1]{}
\usepackage{url}
\usepackage{xcolor}
\usepackage{graphicx}
\usepackage{amsfonts}
\usepackage{amsmath}
\usepackage{amssymb}
\usepackage{amsthm}
\usepackage{multirow}
\usepackage{multicol}
\usepackage{fdsymbol} % \Vbar and \nVbar
\usepackage{algorithm}
\usepackage{algorithmic}
\usepackage{subfigure}

\newtheorem{example}{Example}
\newtheorem{theorem}{Theorem}
\newtheorem{definition}{Definition}
\newtheorem{lemma}{Lemma}

\def\circlerightarrow{\put(2,2.5){\circle{2.5}}\rightarrow}
\def\circleleftarrow{\leftarrow\put(-2,2.5){\circle{2.5}}}
\def\cirlinecir{\put(0.8,2.8){\circle{2.5}}-\put(-0.9,2.8){\circle{2.5}}}

\def\astrightarrow{\put(0.2,-2.2){*}\rightarrow}
\def\astleftarrow{\leftarrow\put(-5,-2.2){*}}

\begin{document}
\title{Towards unique and unbiased causal effect estimation from data with hidden variables}
%%~\IEEEmembership{Fellow,~OSA,}%~\IEEEmembership{Member,~IEEE,}
\author{Debo~Cheng~and~Jiuyong~Li~and~Lin~Liu~and~Kui~Yu~and~Thuc~Duy~Le   and~Jiuxue~Liu%~\IEEEmembership{Life~Fellow,~IEEE}% <-this % stops a space %\IEEEcompsocthanksitem
\IEEEcompsocitemizethanks{\IEEEcompsocthanksitem D. Cheng, J. Li, L. Liu, T.D. Le and J. Liu are with STEM, University of South Australia, Mawson Lakes, South Australia, 5095, Australia. \hfil\break
% note need leading \protect in front of \\ to get a newline within \thanks as \protect
E-mails: chedy055@mymail.unisa.edu.au and Jiuyong.Li@unisa.edu.au
\IEEEcompsocthanksitem K. Yu is with the School of Computer Science and Information Engineering, Hefei University of Technology, Hefei, 230000, China}% <-this % stops an
}
\markboth{}%
{Shell \MakeLowercase{\textit{et al.}}: Bare Demo of IEEEtran.cls for IEEE Journals}



\maketitle

\begin{abstract}
Causal effect estimation from observational data is a crucial but challenging task. Currently, only a limited number of data-driven causal effect estimation methods are available. These methods either provide only a bound estimation of the causal effect of a treatment on the outcome, or generate a unique estimation of the causal effect, but making strong assumptions on data and having low efficiency. In this paper, we identify a practical problem setting and propose an approach to achieving unique and unbiased estimation of causal effects from data with hidden variables. For the approach, we have developed the theorems to support the discovery of the proper covariate sets for confounding adjustment (adjustment sets). Based on the theorems, two algorithms are proposed for finding the proper adjustment sets from data with hidden variables to obtain unbiased and unique causal effect estimation. Experiments with synthetic datasets generated using five benchmark Bayesian networks and four real-world datasets have demonstrated the efficiency and effectiveness of the proposed algorithms, indicating the practicability of the identified problem setting and the potential of the proposed approach in real-world applications.
\end{abstract}


% \fi
\IEEEpeerreviewmaketitle

\section{Introduction}
\label{Sec:Intro}
\IEEEPARstart{C}{ausal} inference~\cite{imbens2015causal,pearl2009causality} has been widely studied to understand the underlying mechanisms of phenomena in economics, medicine, and social science, to name but a few. One major task for causal inference is causal effect estimation, e.g. estimating the effect of a drug on a disease or the effect of a policy on the employment rate of a certain population. \emph{Randomized Controlled Trials} (RCTs) are usually used to estimate causal effects. However, RCTs are often not possible to conduct due to ethical concern, cost, or time constraints.

It is desirable to estimate causal effect using observational data since the collection of observational data is increasing rapidly. Confounding bias is a major obstacle to causal effect estimation from data. To remove confounding bias, covariate adjustment is commonly used, but it is challenging to determine from data the \emph{adjustment set}, i.e. the proper covariates for the adjustment.
Graphical causal modelling provides a theoretical foundation for adjustment set selection~\cite{pearl2009causality,maathuis2015generalized,shpitser2012validity,perkovic2017complete}. When a causal DAG ({Directed Acyclic Graph}) or MAG ({Maximal Ancestral Graph}) representing the causal mechanism is given, the {back-door criterion}~\cite{pearl2009causality} or {generalised back-door criterion}~\cite{maathuis2015generalized} can be used to determine an adjustment set for estimating causal effect from observational data.

However, it is impossible to estimate causal effects from data uniquely without additional constraints. In most real-world applications, the causal graphs are unknown. They must be learned from data, but from data, one cannot discover a unique DAG/MAG. What can be learned from data is an equivalence class of DAGs/MAGs encoding the same conditional independence relationships among variables. This results in uncertainty in the adjustment sets discovered from data and the causal effect estimated using the identified adjustment sets. This is why data-driven causal effect estimation methods often return a set of possible causal effects, i.e. a bound estimation, instead of a unique estimation of the causal effect. For example, a widely used data-driven causal effect estimation method, IDA (Intervention when the DAG is Absent)~\cite{maathuis2009estimating} provides a multiset of estimated causal effects based on data without hidden variables. For data with hidden variables, Hyttinen et al. proposed a bound estimator, CE-SAT, which uses do-calculate~\cite{pearl2009causality} and SAT-based inference to estimate causal effect from data~\cite{hyttinen2015calculus}. Malinsky and Spirtes~\cite{malinsky2017estimating} have extended IDA to LV-IDA (Latent variable IDA), but it is also a bound estimation method. Cheng et al.~\cite{cheng2020causal} proposed a local causal search method (DICE) for determining an adjustment set from data with hidden variables, which provides a bound estimation too.

\begin{table}
\small
\centering
\caption{A summary of representative data-driven causal effect estimators.}
    \begin{tabular}{|c|c|c|}
        \hline
         & \multicolumn{2}{|c|}{\textbf{Hidden variables}} \\
        \cline{2-3}
        \textbf{Estimate} & No & Yes \\
        \hline
       Bound & IDA\cite{maathuis2009estimating} &  CE-SAT\cite{hyttinen2015calculus}; LV-IDA\cite{malinsky2017estimating}; DICE\cite{cheng2020causal} \\
        \hline
       Unique & \emph{CovSel}\cite{haggstrom2018data,de2011covariate};  &    GAC\cite{perkovic2017complete}; EHS\cite{entner2013data}; IDP\cite{jaber2019causal};  \\
       & CFRNet\cite{shalit2017estimating}  &    DAVS (this paper)    \\
        \hline
    \end{tabular}
    \label{tab_methodClassificaiton}
\end{table}

The uncertainty in the results of bound estimation methods hinders the applicability of data-driven causal effect estimation. The most common way to achieve a unique estimation is to use extra constraints to eliminate the uncertainty. H\"{a}ggstr\"{o}m~\cite{haggstrom2018data} developed the Bayesian network methods in conjunction with the Covariate Selection algorithms (CovSel for short)~\cite{de2011covariate} to uniquely estimate the causal effect of a treatment on the outcome, with the assumptions that there are no hidden variables, and all other observed variables are pretreatment variables. With the constraint that the data satisfies the unconfoundedness assumption, Shalit et al.~\cite{shalit2017estimating} have developed a deep learning-based causal defect estimation method, CFRNet. %For data with hidden variables, Entner et al.~\cite{entner2013data} proposed EHS (named after authors' names, Entner, Hoyer, and Spirtes), a method based on conditional independence tests and with the pretreatment variable assumption. %EHS is very inefficient since it performs an exhaustive search over all the variables.

Table~\ref{tab_methodClassificaiton} provides a summary of the above discussed data-driven methods. Our proposed method DAVS is included in the bottom right cell of the table. DAVS deals with data with hidden variables and provides unique causal effect estimate. There are three other methods in the same category, and we focus our discussions on the three methods below.

Among the three methods in the same category as DAVS in Table~\ref{tab_methodClassificaiton}, the GAC (A data-driven method based on the Generalised Adjustment Criterion) method~\cite{perkovic2017complete} and IDP algorithm~\cite{jaber2019causal} are two recently proposed methods to estimate causal effects based on the Markov equivalence class of MAGs represented by a Partial Ancestral Graph (PAG). Both algorithms are proved to be correct and complete based on a PAG. However, in practice, due to insufficient data and/or unreliable statistical tests, the PAG learned from the data may have errors in edge orientation and thus may not correctly represent the underlying MAG. That is, a number of equivalent PAGs can be learned from the data. Fig.~\ref{fig:EquivalencePAG} given an example where PAG$_1$, PAG$_2$ and other PAGs in the same equivalence class are indistinguishable in data. If PAG$_2$ is used by GAC and IDP, confounder $V_3$ will be missed, and hence the causal effect estimation will be biased. An inconsistent PAG causes GAC to give a biased estimation, and this results in low performance of the GAC, as shown in our experiments. IDP aims to infer a probability expression formula based on \emph{do-calculate}\cite{pearl2009causality} applied to a given PAG, and it is a impractical algorithm since the knowledge about the statistical properties of \emph{do-calculate} is still lacking~\cite{vanderweele2009relative,van2019separators}. Furthermore, it also bears the same problem of the GAC algorithm, i.e. the edge orientation errors in the learned PAG.

To eliminate the uncertainty, we propose to use an anchor node, called a Cause or Spouse of the treatment Only (COSO) variable in this paper, such as node $V_1$ in Fig.~\ref{fig:EquivalencePAG}. By using a COSO variable, we will show that the proper adjustment set will be identified from data, and hence the unbiased causal effect is estimated (The detailed example will be described in Section~\ref{subsec:theorems}). The method proposed by Entner et al.~\cite{entner2013data} i.e. EHS (named after authors' surnames, Entner, Hoyer, and Spirtes) also makes use of the same anchor node. However, EHS needs the pretreatment variable assumption, and EHS is very inefficient (as shown in our experiment) since it performs an exhaustive search over all the variables.

\begin{figure}
	\includegraphics[width=8.7 cm]{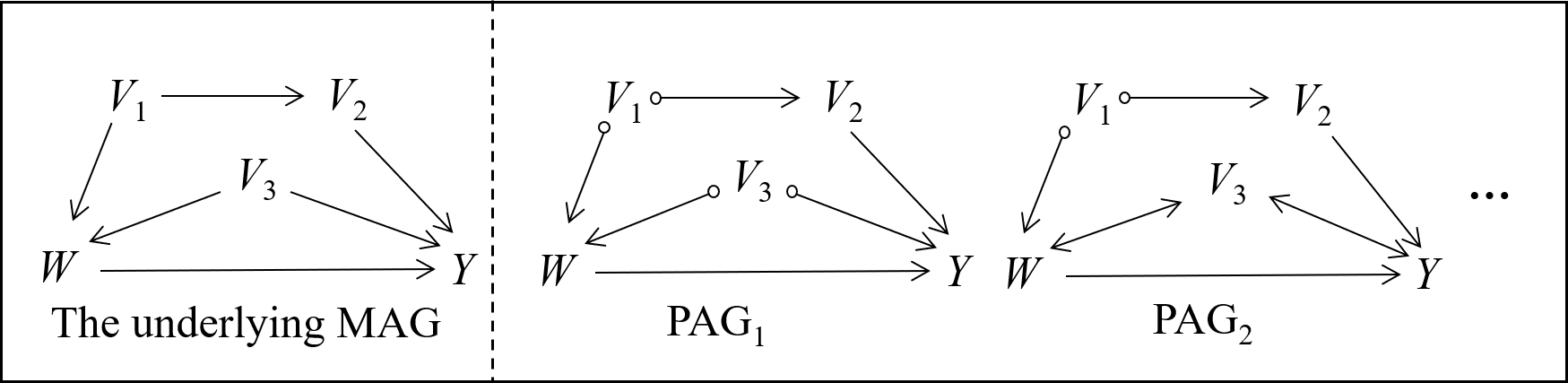}
	\caption{Equivalent PAGs of an underlying MAG in data share the $m$-separations. Set $(V_1, V_3)$, $(V_2, V_3)$ or $(V_1, V_2, V_3)$ is a proper adjustment set. However, if PAG$_2$ is used by GAC or IDP, $V_3$ will be missed in a proper adjustment set and a biased causal effect estimation will be produced. }
	\label{fig:EquivalencePAG}
\end{figure}

In this paper, we aim to develop efficient algorithms and the supporting theory for obtaining unique and unbiased estimation of causal effect from data with hidden variables, under a practical problem setting. The main contributions of the paper are summarized as follows.
\begin{itemize}
  \item We have identified a practical problem setting where causal effects can be estimated uniquely and unbiasedly from data with hidden variables by relaxing the pretreatment variable assumption. We have developed the theorems to support unique and unbiased causal effect estimation from data with hidden variables in the problem setting.
  \item We have developed two algorithms for unique and unbiased causal effect estimation. The developed algorithms are efficient and provide more accurate causal effect estimation than existing methods that deal with hidden variables.
  \item The developed algorithms identify adjustment sets mainly by dependency tests in data directly. This is in contrast to existing methods, which require a correct PAG. From data, our approach mitigates the risk of discovering an inconsistent PAG in a set of equivalent PAGs.
\end{itemize}

\section{Preliminaries and background}
\label{Sec:Pre}
\subsection{Basic definitions and assumptions}
A graph $\mathcal{G}=(\mathbf{V}, \mathbf{E})$ consists of a set of nodes $\mathbf{V}=\{V_{1}, \dots, V_{p}\}$ and a set of edges $\mathbf{E} \subseteq \mathbf{V} \times \mathbf{V}$. In a graph, an edge can be \emph{directed} $\rightarrow$, \emph{bi-directed} $\leftrightarrow$, \emph{non-directed} $\cirlinecir$ or \emph{partially directed} $\circlerightarrow$. We use ``*'' to denote an arbitrary edge mark. For example $\astrightarrow$ can represent $\rightarrow$, $\leftrightarrow$, or $\circlerightarrow$. A \emph{directed graph} consists of only directed edges. A graph which contains only directed or bi-directed edges is a \emph{mixed graph}, and a \emph{partial mixed graph} may contain all types of edges (directed, bi-directed, non-directed and partially directed).

\emph{A path} in a graph is a sequence of distinct adjacent vertices $\langle V_1, \dots, V_k \rangle$ in which each pair of successive vertices is adjacent. On a path $\pi$ $\langle V_1, V_2, \dots, V_k \rangle$ with $k\geq 2$, the vertices $V_1$ and $V_k$ are \emph{endpoints} of $\pi$, and the other rest vertice $V_i$ with $1<i<k$, is a \emph{non-endpoint} vertice of $\pi$. A \emph{causal path} from $V_i$ to $V_j$ is a path on which all edges are directed edges pointing towards $V_j$, and in this case $V_i$ is called an ancestor of $V_i$ or equivalently $V_j$ is a descendant of $V_i$. A \emph{possibly directed path} or \emph{possibly causal path} from $V_i$ to $V_j$ is a path on which none of the edges have an arrowhead pointing towards $V_i$. Here ``Possibly'' implies uncertainty in the direction of the path. For example, $V_i\circlerightarrow V_k \rightarrow V_j$ is a possibly causal path from $V_i$ to $V_j$ because the exact direction of the edge $V_i\circlerightarrow V_k$ is unknown. A path from $V_i$ to $V_j$ that is not possibly causal is called a \emph{non-causal path} from $V_i$ to $V_j$. In a graph, if there exists $V_i \rightarrow V_j$, $V_i$ is a parent of $V_j$ or equivalently $V_j$ is a child of $V_i$. If there is a bi-directed edge between two nodes, e.g. $V_i\leftrightarrow V_j$, $V_i$ is a spouse of $V_j$ or vice versa. When there is a possibly causal path from $V_i$ to $V_j$, $V_i$ is a possible ancestor of $V_j$ and $V_j$ is a possible descendant of $V_i$. In this paper, we use $Pa(V)$,  $Sp(V)$, $An(V)$, $De(V)$, $PossAn(V)$ and $PossDe(V)$ denote the sets of all parents, spouses, ancestors, descendants, possible ancestors, and possible descendants of $V$ respectively.

If there exists $V_k \astrightarrow V_i \astleftarrow V_l$ in a graph, $V_{i}$ is a collider. A \emph{collider path} in a graph is a path on which every non-endpoint node is a collider. A path of length one is a \emph{trivial collider path}. A path $\pi$ $\langle V_i, V_j, V_k\rangle$ with $V_i$ and $V_k$ are (not) adjacent is an $(un)shielded$ $triple$. A path $\pi$ is \emph{unshielded} if all successive triples on $\pi$ are unshielded. A node $V_j$ is a \emph{definite non-collider} on $\pi$ if there exists at least an edge out of $V_j$ on $\pi$, or both edges have a circle mark at $V_j$ and the subpath $\langle V_i, V_j, V_k\rangle$ is an unshielded triple~\cite{zhang2008causal}. A node $V_i$ is named as \emph{definite status} if it is a collider or a definite non-collider on the path. On a path $\pi$, every collider is definite status, and hence, a \emph{definite colldier}. A path $\pi$ is of \emph{definite status} if every node (excluding non-endpoint node) on $\pi$ is of definite status.

\begin{figure}[t]
\centering
    \includegraphics[width=6cm,height=2cm]{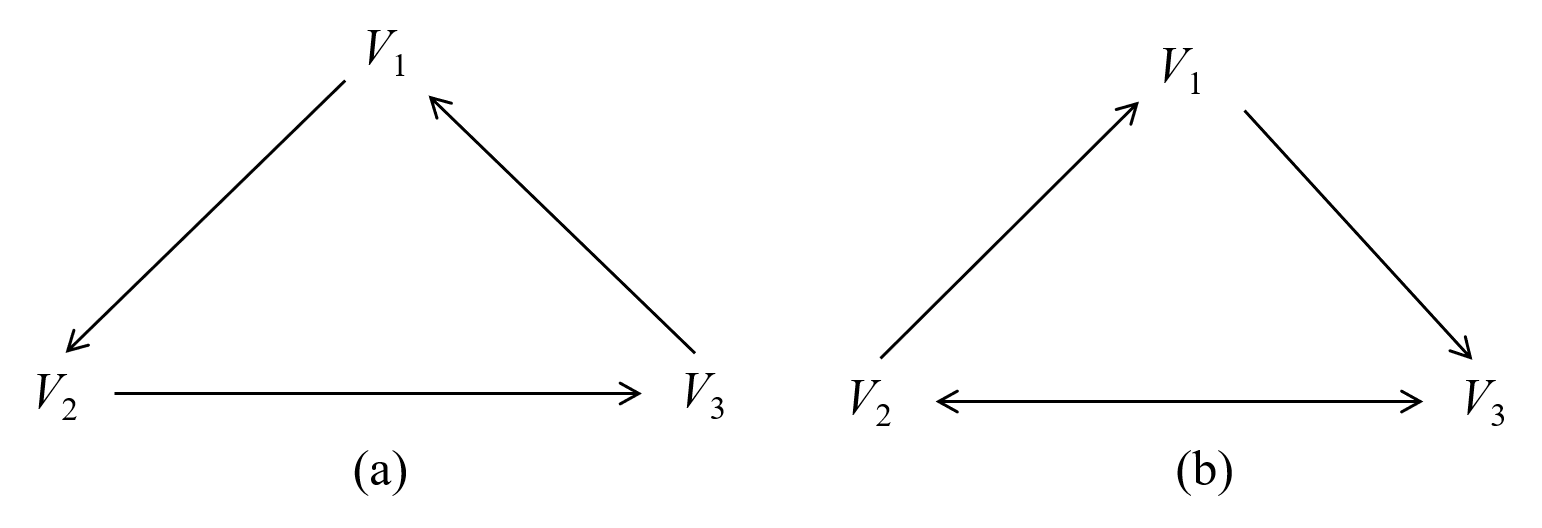}
    \caption{There is a directed cycle between $V_2$ and $V_3$ in (a) and an almost directed cycle between $V_2$ and $V_3$ in (b).}
\label{fig:cycles}
\end{figure}

There is a \emph{directed cycle} between $V_i$ and $V_j$ if $V_i\rightarrow V_j$ and $V_j\in An(V_i)$. There is an \emph{almost directed cycle} between $V_i$ and $V_j$ if $V_i\leftrightarrow V_j$ and $V_j\in An(V_i)$. For example, in Fig.~\ref{fig:cycles} (a), there is a directed cycle between $V_2$ and $V_3$ and in Fig.~\ref{fig:cycles} (b), there is an almost directed cycle between $V_2$ and $V_3$. A \emph{directed acyclic graph} (DAG) contains only directed edges and has no directed cycles. An \emph{ancestral graph}~\cite{richardson2002ancestral} is a mixed graph without directed or almost directed cycles.

In graphical causal modelling, the following condition or assumptions are often involved.

\begin{definition} [Markov condition~\cite{pearl2009causality,richardson2002ancestral}]
\label{def:Markovcondition}
Given a DAG or MAG $\mathcal{G}=(\mathbf{V}, \mathbf{E})$ and $prob(\mathbf{V})$, the joint probability distribution of $\mathbf{V}$, $\mathcal{G}$ satisfies the Markov condition if for $\forall V_i \in \mathbf{V}$, $V_i$ is conditional independent of all non-descendants of $V_i$, given $Pa(V_i)$.
\end{definition}

With a DAG satisfying the Markov condition, the joint distribution of $\mathbf{V}$ can be factorised as $prob (\mathbf{V}) = \prod_i prob(V_i | Pa (V_i))$ by Definition~\ref{def:Markovcondition}.

\begin{definition}[Faithfulness~\cite{Spirtes2000Causation}]
\label{Faithfulness}
A graph $\mathcal{G}=(\mathbf{V}, \mathbf{E})$ is faithful to $prob(\mathbf{V})$ if and only if every independence present in $prob(\mathbf{V})$ is entailed by $\mathcal{G}$ and satisfies the Markov condition. A joint distribution $prob(\mathbf{V})$ is faithful to a graph $\mathcal{G}$ if and only if there exists graph $\mathcal{G}=(\mathbf{V}, \mathbf{E})$ which is faithful to $prob(\mathbf{V})$.
\end{definition}

\begin{definition}[Causal sufficiency~\cite{Spirtes2000Causation}]
A given dataset satisfies causal sufficiency if in the dataset for every pair of observed variables, all their common causes are observed.
\end{definition}

In data, the assumption of causal sufficiency is often unwarranted since there is not a closed-world. \emph{Ancestral graphs} are usually used to represent the underlying causal mechanism of the data which may involve hidden variables (while DAGs are often used to represent causal mechanisms without hidden variables).

The concept of m-separation~\cite{richardson2002ancestral} is used to link an ancestral graph with the conditional dependencies/independencies in a probabilistic distribution.

\begin{definition}[m-separation~\cite{richardson2002ancestral}]
\label{m-separation}
In an ancestral graph $\mathcal{G}=(\mathbf{V}, \mathbf{E})$, a path $\pi$ between $V_{i}$ and $V_{j}$ is said to be m-separated by a set of nodes $\mathbf{Z}\subseteq \mathbf{V}\setminus\{V_i, V_j\}$ (possibly $\emptyset$) if every non-collider on $\pi$ is a member of $\mathbf{Z}$; and every collider on $\pi$ is not in $\mathbf{Z}$ and none of the descendants of the colliders is in $\mathbf{Z}$.
Two nodes $V_{i}$ and $V_{j}$ are said to be m-connected by $\mathbf{Z}$ in $\mathcal{G}$ if $V_{i}$ and $V_{j}$ are not m-separated by $\mathbf{Z}$.
\end{definition}

When two ancestral graphs represent the same m-separations among the observed variables $\mathbf{V}$, they are called \emph{Markov equivalent} as defined below.

\begin{definition}[Markov equivalent~\cite{ali2009markov}]
  \label{markovequ}
Two ancestral graphs $\mathcal{G}_1$ and $\mathcal{G}_2$ with the same set of nodes are said to be \emph{Markov equivalent}, denoted $\mathcal{G}_1 \sim \mathcal{G}_2$, if for any three disjoint sets $\mathbf{V}_1$, $\mathbf{V}_2$ $\mathbf{V}_3$ ($\mathbf{V}_1$, $\mathbf{V}_2$ not empty), $\mathbf{V}_1$ and $\mathbf{V}_2$ are m-separated given $\mathbf{V}_3$ in $\mathcal{G}_1$ if and only if $\mathbf{V}_1$ and $\mathbf{V}_2$ are m-separated given $\mathbf{V}_3$ in $\mathcal{G}_2$.
\end{definition}

The set of all ancestral graphs that encode the same set of m-separations forms a \emph{Markov equivalence class}~\cite{ali2009markov}.

\begin{definition}[MAG~\cite{richardson2002ancestral}]
 \label{MAG}
An ancestral graph $\mathcal{G}=(\mathbf{V}, \mathbf{E})$ is referred to as a maximal ancestral graph (MAG) when for every pair of non-adjacent nodes $V_i$ and $V_j$ in $\mathbf{V}$, there exists $\mathbf{Z}\subseteq \mathbf{V}\setminus\{V_i, V_j\}$ such that $V_i$ and $V_j$ are m-separated by $\mathbf{Z}$.
\end{definition}

A DAG has obviously met both conditions in~\ref{MAG}, so syntactically, a DAG is also a MAG without bi-directed edges~\cite{zhang2008causal}. A set of Markov equivalent MAGs can be represented uniquely by a \emph{partial ancestral graph} (PAG).

\begin{definition}[PAG~\cite{zhang2008causal}]
\label{def:pag}
Let $[\mathcal{G}]$ be the Markov equivalence class of a MAG $\mathcal{G}$. The PAG $\mathcal{G}$ for $[\mathcal{G}]$ is a partial mixed graph such that (I). $\mathcal{G}$ has the same adjacent relations among nodes as $\mathcal{G}$ does; (II). For an edge, its a mark of arrowhead or mark of tail is in $\mathcal{G}$ if and only if the same mark of arrowhead or the same mark of tail is shared by all MAGs in $[\mathcal{G}]$.
\end{definition}

\subsection{Confounding adjustment}
\label{subsec:causal}
Let $\mathcal{G} = (\mathbf{V}, \mathbf{E})$ be a graph, and $\mathbf{V}=\{W, Y\}\cup \mathbf{X}$, where $W$ is the treatment variable of interest, $Y$ the outcome variable, and $\mathbf{X}$ the set of all other observed variables. We would like to query a dataset for the causal effect of $W$ on $Y$ from data. In the corresponding causal diagram, there is at least one causal path from $W$ to $Y$. In this paper, we are interested in estimating the average causal effect of $W$ on $Y$, as defined below.

\begin{definition}[Average causal effect]
\label{def:ATC}
The Average Causal Effect of $W$ on $Y$ is defined as $ACE(W, Y)=\mathbf{E}(Y|do~(W=1))-\mathbf{E}(Y|do~(W=0))$, where $do(W=w)$ is the $do$-operator indicating the manipulation of $W$ by setting it to the value $w$~\cite{pearl2009causality}.
\end{definition}

Given a proper adjustment set $\mathbf{Z}\subseteq \mathbf{X}$, $ACE(W, Y)$ can be estimated unbiasedly by
\begin{equation}
\label{eq:eq01}
\begin{aligned}
ACE(W,Y)&=\sum_{\mathbf{z}}[\mathbf{E}(Y|W=1,\mathbf{Z}=\mathbf{z})-\\
&\mathbf{E}(Y|W=0,\mathbf{Z}=\mathbf{z})]prob(\mathbf{Z}=\mathbf{z})
\end{aligned}
\end{equation}
When there are no hidden variables, and the causal DAG representing the causal mechanism generating the data is known, the back-door criterion~\cite{pearl2009causality} can be used to identify an adjustment set from the DAG. When there are hidden variables, the causal mechanism generating the data can be represented by a MAG, and the generalized adjustment criterion~\cite{perkovic2017complete} can be used to identify an adjustment set based on the MAG (or a PAG). To introduce the generalised adjustment criterion, we firstly present the following relevant concepts as below.

\begin{definition}[Visibility~\cite{zhang2008causal}]
 \label{Visibility}
Given a PAG or MAG $\mathcal{G}$, a directed edge $X_{i}\rightarrow X_{j}$ in $\mathcal{G}$ is \textbf{visible} if there is a node $X_{k}$ not adjacent to $X_{j}$, such that either there is an edge between $X_{k}$ and $X_{i}$ that is into $X_{i}$, or there is a collider path between $X_{k}$ and $X_{i}$ that is into $X_{i}$ and every node on this path is a parent of $X_{j}$. Otherwise, $X_{i}\rightarrow X_{j}$ is said to be \textbf{invisible}.
\end{definition}

Fig.~\ref{fig:visibleedge} shows two possible structures that can result in a visible edge. As mentioned earlier, we only consider query the effect of a single treatment variable $W$ on a single response variable $Y$ of interest in the whole paper. Hence, the following definitions are introduced under the single treatment $W$ and response $Y$, and the original definitions on the set of $\mathbf{W}$ and $\mathbf{Y}$ can be found in the literature~\cite{maathuis2015generalized,perkovic2017complete,van2014constructing}.

\begin{figure}
\centering
    \includegraphics[width=5cm,height=2cm]{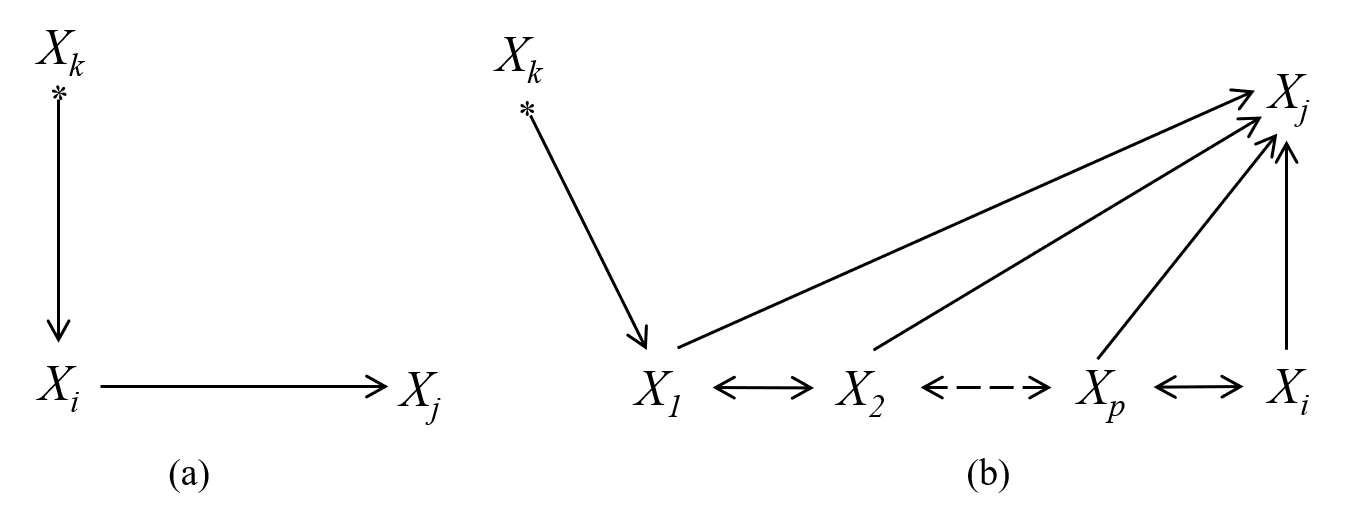}
    \caption{Two possible structures in MAGs and PAGs for a visible edge $X_i\rightarrow X_j$~\cite{zhang2008causal}. Nodes $X_k$ and $X_j$ must be nonadjacent.}
\label{fig:visibleedge}
\end{figure}

%To introduce the generalized adjustment criterion, we need to define amenability~\cite{van2014constructing}.

\begin{definition}[Amenability~\cite{van2014constructing}]
\label{Amenability}
Given a PAG or MAG $\mathcal{G}$ and a pair of variables $(W, Y)$ in $\mathcal{G}$, $\mathcal{G}$ is adjustment amenable w.r.t. $(W, Y)$ if each possibly directed path from $W$ to $Y$ in $\mathcal{G}$ starts with a visible edge out of $W$.
\end{definition}

\begin{definition}[Forbidden set $Forb(W, Y, \mathcal{G})$]
\label{def:forbiden}
Given a PAG or MAG $\mathcal{G}$ and a pair of variables ($W, Y$) in $\mathcal{G}$, the forbidden set w.r.t. ($W, Y$) is defined as
$Forb(W, Y, \mathcal{G})=\{X\in \mathbf{V}: X\in PossDe(W)$, and $X$ lies on a possibly directed path from $W$ to $Y$ in $\mathcal{G}\}$.
\end{definition}

The forbidden set in $\mathcal{G}$ contains the variables that cannot be included in an adjustment set. Following~\cite{colombo2012learning}, for describing a small number of candidate adjustment sets, we need to define possible d-separation set, denoted as $pds(W, Y, \mathcal{G})$.

\begin{definition}[$pds(W, Y, \mathcal{G})$~\cite{colombo2012learning}]
\label{def:pds}
Given a PAG or MAG $\mathcal{G}$ and a pair of variables $(W, Y)$ in $\mathcal{G}$, $X\in pds(W, Y, \mathcal{G})$ if and only if there is a path $\pi$ between $X$ and $W$ in $\mathcal{G}$ such that for every subpath $<X_l, X_k, X_h>$ on $\pi$, either $X_k$ is a collider or $<X_l, X_k, X_h>$ is a triangle, i.e. each pair of nodes in the triple are adjacent.
\end{definition}

\begin{definition}[Generalised back-door path~\cite{maathuis2015generalized}]
\label{generalised back door path}
Given a MAG or PAG $\mathcal{G}$ and a pair of variables ($W, Y$) in $\mathcal{G}$, a path between $W$ and $Y$ is a generalised back-door path from $W$ to $Y$ if it does not have a visible edge out of $W$.
\end{definition}

Now, we introduce the generalised adjustment criterion, which is a sufficient and complete graphical adjustment criterion for identifying an adjustment set from a MAG or PAG~\cite{perkovic2017complete}.

\begin{definition}[Generalised adjustment criterion~\cite{perkovic2017complete}]
\label{def:GAC}
Given a MAG or PAG $\mathcal{G}$ with a pair of variables $(W, Y)$ and a set of nodes $\mathbf{Z}$ in $\mathcal{G}$. Then $\mathbf{Z}$ satisfies the generalised adjustment criterion relative to $(W, Y)$ if (I). $\mathcal{G}$ is adjustment amenable relative to $(W, Y)$, (II). $\mathbf{Z}\cap Forb(W, Y, \mathcal{G})= \emptyset$, and (III). all definite status non-causal paths between $W$ and $Y$ are blocked by $\mathbf{Z}$ in $\mathcal{G}$.
\end{definition}

Note that the method developed by Perkovi\'{c} et al. based on the generalised adjustment criterion is also referred as GAC, but it is an algorithm based on the generalised adjustment criterion and itself is not the criterion.

\section{Causal effect estimation from data with hidden variables}
\label{Sec:Alg}
In this section, we develop the data-driven method for unique and unbiased causal effect estimation from data with hidden variables, including the problem setting, i.e. the identified case where the causal effect can be uniquely and unbiasedly estimated, the theorems, and the algorithms.

\subsection{Problem setting}
As mentioned previously, from observational data, what we can learn is an equivalence class of DAGs/MAGs. Hence the problem of causal effect estimation does not have a unique solution in general. To obtain a unique estimation of a causal effect, assumptions are needed. In this paper, we have identified a practical case where a causal effect can be uniquely estimated from data with hidden variables. Specifically, we assume that in the equivalence class of the MAGs learned (represented by a PAG), there exists a Cause Or Spouse of the treatment Only (COSO) variable, as defined below.

\begin{definition}[COSO variable]
\label{def:COSO}
Let $\mathcal{G}$ be the PAG learned from a dataset with treatment $W$, outcome $Y$ and the set of other observed variables $\mathbf{X}$. A variable $Q\in\mathbf{X}$, is a COSO variable if $Q\in (Pa(W)\cup Sp(W))$ and $Q\notin (Pa(Y)\cup Sp(Y))$.
\end{definition}

A COSO variable can be easily found in many applications. For example, when studying the effect of smoking on lung cancer~\cite{Spirtes2000Causation}, family influence causes smoking but does not cause lung cancer directly. Family influence is a COSO variable w.r.t. (smoking, lung cancer). In the study of the effect of job training on income, marriage impacts job training but does not directly affect income~\cite{lalonde1986evaluating}. Marriage is a COSO variable for (job training, income). Compared to the existing data-driven methods for unique and unbiased causal effect estimation~\cite{entner2013data,de2011covariate,shalit2017estimating,haggstrom2018data}, which require that $\mathbf{X}$ contains only pretreatment variables or satisfies the unconfoundedness assumption (as discussed in the Introduction~\ref{Sec:Intro}), our COSO variable assumption is more practical and enables broader real-world applications. Compared to GAC which read off an adjustment set $\mathbf{Z}$ from the learned PAG since $W\Vbar Y\mid \mathbf{Z}$ can not be tested from data directly, using a COSO variable to eliminate (if COSO is given by the user) or reduce (local discovery from data) the uncertainly by identifying an adjustment set $\mathbf{Z}$ from data directly as described below.

\subsection{The theorems}
\label{subsec:theorems}
We would like to estimate the causal effect a treatment of interest, $W$ on $Y$, the outcome from data with hidden variables, and thus we consider there is at least one causal path from $W$ to $Y$ in the corresponding causal diagram, since otherwise, $W$ would not be able to have causal effect on $Y$~\cite{pearl2009causality}. We set out to come up with our theorems to identify adjustment sets by conditional independence tests. For proving our theorems, we first propose a lemma as below. %\textcolor{red}{Moreover, in our problem setting, we provide the following lemma to show that definite status non-causal paths between $W$ and $Y$ can be reduced to generalised back-door paths between $W$ and $Y$.}

\begin{lemma}
\label{lemma}
    Given a MAG or PAG $\mathcal{G}$ with a pair of variables $(W, Y)$ and a set of nodes $\mathbf{Z}$ in $\mathcal{G}$. $\mathcal{G}$ is adjustment amenable relative to $(W, Y)$ and $\mathbf{Z}\cap Forb(W, Y, \mathcal{G})= \emptyset$. If all generalised back-door paths between $W$ and $Y$ are blocked by $\mathbf{Z}$ in $\mathcal{G}$, then $\mathbf{Z}$ blocks all definite status non-causal paths between $W$ and $Y$ in $\mathcal{G}$.
\end{lemma}
\begin{proof}
    All definite status non-causal paths between $W$ and $Y$ in Definition~\ref{def:GAC}, are either the generalised back-door paths between $W$ and $Y$, or $W\circlerightarrow \dots \astrightarrow C \astleftarrow \dots Y$. The latter is natural blocked by $\emptyset$ since $C$ is a collider on the path. There are two situations for the latter case, (1). if the path $\pi$ is form of $W\circlerightarrow \dots \circlerightarrow C \astleftarrow \dots Y$, then $C\in Forb(W, Y, \mathcal{G})$ such that $C\notin \mathbf{Z}$  because $\mathbf{Z}\cap Forb(W, Y, \mathcal{G})= \emptyset$. (2). if the path $\pi$ is form of $W\circlerightarrow \dots \circlerightarrow X_k \leftrightarrow C \astleftarrow \dots Y$ and $C\in \mathbf{Z}$, then $X_k\notin \mathbf{Z}$ since $X_k \in Forb(W, Y, \mathcal{G})$.

    Therefore, if all generalised back-door paths between $W$ and $Y$ are blocked by the set $\mathbf{Z}$ in $\mathcal{G}$, then $\mathbf{Z}$ satisfies the blocking condition in the generalised adjustment criterion.
\end{proof}

Now we present the following theorem presents the condition for searching for an adjustment set in data with hidden variables under our problem setting.

\begin{theorem}
\label{theo:theo01}
Given a PAG $\mathcal{G}$ containing a pair of variables $(W, Y)$ and the other variables $\mathbf{X}$, and assume that there exists a COSO variable $Q\in\mathbf{X}$. There exists a set of variables $\mathbf{Z} \subseteq\mathbf{X}\setminus (\{Q\} \cup Forb(W, Y, \mathcal{G}))$, then $\mathbf{Z}$ is an adjustment set if and only if $Q\Vbar Y|\mathbf{Z}\cup \{W\}$.
\end{theorem}
\begin{proof}
As $Q$ is a COSO variable, $\mathcal{G}$ must contain an edge $Q\astrightarrow W$, i.e. there is an edge from $Q$ into $W$. As there is at least a causal path from $W$ to $Y$, there must be an edge out of $W$. Then by Definition~\ref{Visibility}, when there is an edge out of $W$, the edge must be visible. Thus, the PAG $\mathcal{G}$ is adjustment amenable w.r.t. $(W, Y)$ by Definition~\ref{Amenability}, i.e. the condition (I) of the generalised adjustment criterion (Definition~\ref{def:GAC}) holds.

The condition $\mathbf{Z} \subseteq\mathbf{X}\setminus (\{Q\} \cup Forb(W, Y, \mathcal{G}))$ is to remove all possible descendants of $W$ in $\mathcal{G}$. Hence, the set $\mathbf{Z}$ satisfies the condition (II) of the generalised adjustment criterion (Definition~\ref{def:GAC}).

Now we first prove the sufficient condition, i.e. if $Q\Vbar Y|\mathbf{Z}\cup\{W\}$, then $\mathbf{Z}$ blocks all definite status non-causal paths, i.e. condition (III) of the generalised adjustment criterion (Definition~\ref{def:GAC}) is satisfied, and thus $\mathbf{Z}$ is an adjustment set given that conditions (I) and (II) of the generalised adjustment criterion have been proved above. We prove this using contradiction. Suppose that a generalised back-door path $\pi$ from $W$ to $Y$ is not blocked by $\mathbf{Z}$. As $\pi$ must be in the form of $W\astleftarrow \dots \astrightarrow Y$ according to Definition~\ref{generalised  back door path}. Moreover, the edge $Q\astrightarrow W$ and $\pi$ will collide at $W$ (i.e. $W$ is a collider), so $Q$ is m-connected to $Y$ given $\mathbf{Z}\cup \{W\}$. This violates $Q\Vbar Y|\mathbf{Z}\cup \{W\}$. Therefore, all generalised back-door paths are blocked by $\mathbf{Z}$, and thus condition (III) of the generalised adjustment criterion holds according to Lemma~\ref{lemma}.

Next, we prove the necessary condition, i.e. if $\mathbf{Z} \subseteq\mathbf{X}\setminus (\{Q\} \cup Forb(W, Y, \mathcal{G}))$ is an adjustment set, then $Q\Vbar Y|\mathbf{Z}\cup \{W\}$. When $\mathbf{Z}$ is an adjustment set, then $\mathbf{Z}$ blocks all definite status non-causal paths from $W$ to $Y$ by condition (III) of the generalised adjustment criterion (Definition~\ref{def:GAC}). Hence, by Lemma~\ref{lemma}, $\mathbf{Z}$ blocks all generalised back-door paths between $W$ and $Y$. Moreover, the paths from $Q$ to $Y$ which are form of $Q \astrightarrow W$ and the causal paths between $W$ and $Y$ (as they must be in the form of $Q \astrightarrow W \rightarrow \dots \rightarrow Y$) are blocked by $W$. Thus, $Q$ and $Y$ are m-separated given $\mathbf{Z}\cup \{W\}$ in PAG $\mathcal{G}$, i.e. $Q\Vbar Y|\mathbf{Z}\cup \{W\}$.
\end{proof}

We use Fig.~\ref{fig:EquivalencePAG} as an example to illustrate Theorem~\ref{theo:theo01}.
Assume that the learned PAG is PAG$_2$ in Fig.~\ref{fig:EquivalencePAG} (which is inconsistent with the underlying MAG). $Forb(W, Y, \mathcal{G}) = \emptyset$ and hence $\{V_1, V_2, V_3\}$ and its subsets are candidate adjustment sets. The GAC reads adjustment sets from PAG$_2$, and obtains the adjustment sets $\{V_1\}$, $\{V_2\}$ or $\{V_1, V_2\}$ since $V_3$ is a collider in PAG$_2$. Theorem~\ref{theo:theo01} uses PAG$_2$ only for finding $Forb(W, Y, \mathcal{G})$, and uses the criterion, $V_1 \Vbar Y | \mathbf{Z} \cup \{W\}$, for identifying adjustment sets where $V_1$ is the COSO. Based on the criterion, $\{V_2, V_3\}$ is identified as the adjustment set and this adjustment is correct for the underlying MAG. Theorem~\ref{theo:theo01} does not make the same mistake as the GAC when a PAG is misspecified. Assume that the underlying MAG is consistent with PAG$_2$. The GAC identifies the adjustment sets correctly. Theorem~\ref{theo:theo01} will identify $\{V_2\}$ as the adjustment set and this adjustment set is also correct.

To reduce the search space for $\mathbf{Z}$ by Theorem~\ref{theo:theo01}, we develop the following theorem.

\begin{theorem}
    \label{theo:theo02}
    Given a PAG $\mathcal{G}$ containing a pair of variables $(W, Y)$ and the other variables $\mathbf{X}$, and assume that there exists a COSO variable $Q\in\mathbf{X}$. There exists a set of variables $\mathbf{Z} \subseteq pds(W, Y, \mathcal{G})\setminus (\{Q\} \cup Forb(W, Y, \mathcal{G}))$, then $\mathbf{Z}$ is an adjustment set if and only if $Q\Vbar Y|\mathbf{Z}\cup \{W\}$.
\end{theorem}
\begin{proof}
$pds(W, Y, \mathcal{G})$ is a subset of $\mathbf{X}$ and from~\cite{malinsky2017estimating}, $pds(W, Y, \mathcal{G})$ contains all nodes which lie on the generalised back-door paths between $W$ and $Y$. Hence, the search for a valid adjustment set can be done only in $pds(W, Y, \mathcal{G})\setminus (\{Q\}\cup Forb(W, Y, \mathcal{G}))$ instead of in $\mathbf{X} \setminus (\{Q\}\cup Forb(W, Y, \mathcal{G}))$. Then the proof follows from the proof of Theorem~\ref{theo:theo01}.%all minimal adjustment sets, so the search for a minimal adjustment set can be done only in $pds(W, Y, \mathcal{G})\setminus (\{Q\}\cup Forb(W, Y, \mathcal{G}))$ instead of in $\mathbf{X} \setminus (\{Q\}\cup Forb(W, Y, \mathcal{G}))$. Then the proof follows from the proof of Theorem~\ref{theo:theo01}.
\end{proof}

\begin{figure}
\centering
    \includegraphics[width=4.5cm,height=1.8cm]{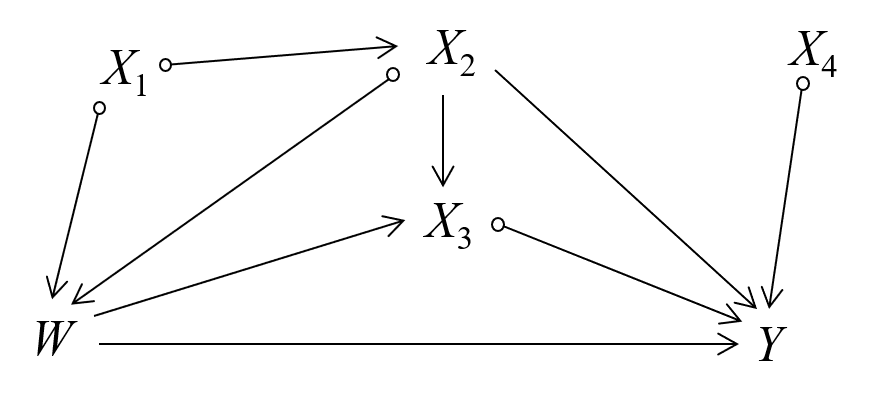}
    \caption{An example PAG, used in Example~\ref{Example:01}.}
\label{fig:pag}
\end{figure}

The following example illustrates Theorems~\ref{theo:theo01} and~\ref{theo:theo02}.

\begin{example}
    \label{Example:01}
Refer to the PAG in Fig.~\ref{fig:pag}, denoted as $\mathcal{G}$, which contains a pair of variables $(W, Y)$ and has $X_1$ as a COSO variable. From $\mathcal{G}$, $Forb(W, Y, \mathcal{G})=\{X_3, Y\}$ and $pds(W, Y, \mathcal{G})=\{X_1, X_2, X_3, Y\}$.

Considering Theorem~\ref{theo:theo01}, as $\mathbf{X}\setminus (\{Q\} \cup Forb(W, Y, \mathcal{G})) = \{X_2, X_4\}$, and $\mathbf{Z}$ can be a subset of $\{X_2, X_4\}$, i.e. $\{X_2\}$, $\{X_4\}$ or $\{X_2, X_4\}$. From $\mathcal{G}$, $X_1\Vbar Y|\{X_2,W\}$ and $X_1\Vbar Y|\{X_2, X_4, W\}$, so following Theorem~\ref{theo:theo01}, $\{X_2\}$ and $\{X_2, X_4\}$ are adjustment sets for unbiased estimation of $ACE(W,Y)$. According to $\mathcal{G}$, we actually can see that $X_2$ blocks all the four generalised back-door paths from $W$ to $Y$: $W \circleleftarrow X_1 \circlerightarrow X_2 \rightarrow Y$; $W\circleleftarrow X_1 \circlerightarrow X_2\rightarrow X_3 \circlerightarrow Y$; $W\circleleftarrow X_2\rightarrow X_3 \circlerightarrow Y$ and $W\circleleftarrow X_2 \rightarrow Y$.

Considering Theorem~\ref{theo:theo02}, as $pds(W, Y, \mathcal{G})\setminus (\{Q\}\cup Forb(W, Y, \mathcal{G})) = \{X_2\}$, the search space for $\mathbf{Z}$ reduces to $\{X_2\}$, instead of $\{X_2, X_4\}$ when following Theorem~\ref{theo:theo01}, and we obtain ${X_2}$ as an adjustment set by Theorem~\ref{theo:theo02}.
\end{example}
%and we obtain $X_2$ as an adjustment set by Theorem~\ref{theo:theo02},

\begin{theorem}
  \label{theo:theo03}
Given a PAG $\mathcal{G}$ containing a pair of variables $(W, Y)$ and the other variables $\mathbf{X}$, and assume that there exists a COSO variable $Q\in\mathbf{X}$. Let $\mathcal{Z}$ be the set of all adjustment sets found following Theorem~\ref{theo:theo01} or Theorem~\ref{theo:theo02}, the estimated value of $ACE(W,Y)$ by adjusting $\mathbf{Z}$ is the same, i.e. unique $\forall \mathbf{Z}\in\mathcal{Z}$.
\end{theorem}
\begin{proof}
Any adjustment set $\mathbf{Z}$ found by following Theorem~\ref{theo:theo01} or Theorem~\ref{theo:theo02} is a proper adjustment set for the underlying MAG of the given PAG $\mathcal{G}$ since $Q\Vbar Y \mid \{W\}\cup \mathbf{Z}$ holds. Hence, the causal effect estimated by adjusting any $\mathbf{Z}$ is unbiased. In other words, the estimated causal effect is always the same when adjusting any $\mathbf{Z}$ found with Theorem~\ref{theo:theo01} or Theorem~\ref{theo:theo02}.
\end{proof}

Theorem~\ref{theo:theo03} shows that the estimated causal effect is unique when adjusting on any of the adjustment sets found with Theorem~\ref{theo:theo01} or Theorem~\ref{theo:theo02}. This is a significant advantage of our proposed method over most existing data-driven methods which listed in Table~\ref{tab_methodClassificaiton}, such as LV-IDA \cite{malinsky2017estimating}, which provides a bound or a multiset of estimated causal effects, and GAC~\cite{perkovic2017complete} read off an adjustment set from the learned PAG which suffers from the errors of the edge orientations. The invalid adjustment set is identified because $W\Vbar Y\mid \mathbf{Z}$ is not testable from data as $W$ is a cause of $Y$. It is also worth noting again that Theorems~\ref{theo:theo01} and~\ref{theo:theo02} only assume the existence of a COSO variable in the data, which is more practical and less restricted than the pretreatment variable assumption.

Based on Theorem~\ref{theo:theo03}, adjusting any $\mathbf{Z}$ satisfying Theorem~\ref{theo:theo01} or Theorem~\ref{theo:theo02} will give an unbiased causal effect estimation. In a dataset, there will be variance associated with each estimation. In our following algorithms, we will employ multiple estimations based on all minimal adjustment sets for estimation to reduce variance.

\subsection{The proposed algorithms}
\label{subsec:dase}
In this section, based on Theorem~\ref{theo:theo02} and Theorem~\ref{theo:theo03}, we propose two Data-driven Adjustment Variable Selection algorithms, DAVS-$Q$ (Algorithm~\ref{pseudocode01}) and DAVS (Algorithm~\ref{pseudocode02}) for unique and unbiased causal effect estimation from data with hidden variables. As mentioned above, it is desirable to use multiple adjustment sets for robustness. Thus, \emph{Apriori pruning strategy}~\cite{han2011data} is used to search for all minimal adjustment sets from data. The correctness of the search strategy is guaranteed by the principle that any superset of a minimal adjustment set must not be a minimal adjustment set.

Both algorithms take as input a dataset and a PAG can be learned from the dataset using rFCI~\cite{colombo2012learning} in this work. DAVS-$Q$ requires a given COSO variable, while DAVS finds COSO variables from data. In many applications, a COSO variable is known based on domain knowledge, then DAVS-$Q$ is the better choice. When a COSO variable is unknown, users can employ DAVS, which finds the set of candidate COSO variables, $\mathbf{Q}$. For each variable $Q\in \mathbf{Q}$, DAVS calls DAVS-$Q$ to estimate $ACE(W, Y)$.

So the two algorithms take the same search strategy, \emph{Apriori pruning strategy}~\cite{han2011data}, and estimate causal effect using a found adjustment set (see Lines 11-22 of Algorithm~\ref{pseudocode01}). Furthermore, the algorithms use a bottom-up search (from single variables to multiple variables) for efficient utilisation of data in conditional independence tests. Note that an empty set is also a legitimate adjustment set if it satisfies Theorem~\ref{theo:theo02}. This situation is checked and dealt with in Line 3 of Algorithm~\ref{pseudocode01}. The pseudocode of \emph{Apriori pruning strategy} is provided in Algorithm~\ref{pse:Congen}, which is used in Line 21 of Algorithm~\ref{pseudocode01} to generate level $k$ candidate adjustment sets (stored in $C_k$). We use a \emph{hash table} to store candidates to save search time (See Line 2 in Algorithm~\ref{pse:Congen}). All indexes of the items are sorted in lexicon order for fast pruning. The first $k-2$ indexes of all sets in $\mathbf{C}_{k-1}$ are hashed for generating new candidates in $\mathbf{C}_{k}$ (See Line 2). Line 4 combines a pair of sets with the same first $k-2$ indexes to create a new candidate with $k$ indexes of variables. Lines 5 to 9 aim to prune candidate set $\mathbf{C}_{k}$ by checking all subsets of a candidate to ensure the candidate is potentially minimal.

In our implementation, the input PAG $\mathcal{G}$ is learned using rFCI~\cite{colombo2012learning} implemented in the $\mathbb{R}$ package \emph{pcalg}~\cite{kalisch2012causal}. The conditional independence tests are conducted by using \emph{GaussianCItest} and \emph{binCItest} in \emph{pcalg} for Gaussian and binary datasets respectively. The calculation of $ACE(W,Y)$ following Eq.(\ref{eq:eq01}) is done by \emph{lm} in $\mathbb{R}$ package \emph{stats}~\cite{maathuis2009estimating,malinsky2017estimating} and \emph{stdGlm} in $\mathbb{R}$ package \emph{stdReg}~\cite{sjolander2016regression} for continuous and binary outcome $Y$ respectively.

\begin{algorithm}[t]
\caption{Data-driven Adjustment Variable Selection given $Q$ (DAVS-$Q$)}
\label{pseudocode01}
\noindent {\textbf{Input}}: Dataset $\mathbf{D}$ with $W, Y$ and $\mathbf{X}$, $Q\in \mathbf{X}$, PAG $\mathcal{G}$ learned from $\mathbf{D}$.\\
\noindent {\textbf{Output}}: $\mathbf{\Psi}$.
\begin{algorithmic}[1]
\STATE{Let $\mathbf{\Psi} = NULL$}
\IF {$Q\Vbar Y|\mathbf{W}$}
\STATE{ Calculate $ACE(W, Y)$ via Eq.(\ref{eq:eq01}) given $\mathbf{Z} =\emptyset$ }
\ELSE
\IF{$\mathcal{G}$ is not given}
\STATE{Call the function rFCI to learn PAG $\mathcal{G}$ from $\mathbf{D}$.}
\ELSE
\STATE {Obtain $pds(W, Y, \mathcal{G})$ and $Forb(W, Y, \mathcal{G})$ from $\mathcal{G}$}
\STATE {$\mathbf{C}_1 = pds(W, Y, \mathcal{G})\setminus (\{Q\} \cup Forb(W, Y, \mathcal{G}))$}
\STATE {$k = 1$}
\WHILE{$\mathbf{C}_{k}\neq \emptyset$}
\FOR {each element $\mathbf{Z}\in \mathbf{C}_{k}$}
\IF{$\mathbf{Z}$ is not in $\mathbf{\Psi}$}
\IF{$Q\Vbar Y|\mathbf{Z}\cup \{W\}$}
\STATE{Calculate $ACE(W, Y)$ via Eq.(\ref{eq:eq01}) given $\mathbf{Z}$ and remove $\mathbf{Z}$ from $\mathbf{C}_{k}$}
\STATE{$\mathbf{\Psi}\leftarrow \cup(\mathbf{Z}, ACE(W, Y))$}
\ENDIF
\ENDIF
\ENDFOR
\STATE{$k = k+1$}
\STATE{/*Call the apriori pruning function in Algorithm~\ref{pse:Congen}*/ \\$\mathbf{C}_{k} = Candidate.gen(\mathbf{C}_{k-1}, k)$ }
\ENDWHILE
\ENDIF
\ENDIF
\STATE{Calculate the mean of all estimated $ACE(W, Y)$ in $\mathbf{\Psi}$, and stored the calculated value in $\mathbf{\Psi}$.}
\RETURN{$\mathbf{\Psi}$.}
\end{algorithmic}

\end{algorithm}
\begin{algorithm}[t]
\caption{Data-driven Adjustment Variable Selection without given $Q$ (DAVS)}
\label{pseudocode02}
\noindent {\textbf{Input}}: Dataset $\mathbf{D}$ with $W, Y$ and $\mathbf{X}$.\\
\noindent {\textbf{Output}}: $\mathbf{\Psi}$.
\begin{algorithmic}[1]
\STATE {Let $\mathbf{\Psi} =  NULL$}
\STATE{Call the function rFCI to learn PAG $\mathcal{G}$ from $\mathbf{D}$.}
\STATE {Find $Pa(W), Sp(W)$, $Pa(Y)$ and $Sp(Y)$ from $\mathcal{G}$}
\STATE {$\mathbf{Q}= (Pa(W)\cup Sp(W))\setminus (Pa(Y)\cup Sp(Y))$}
\FOR {each $Q\in \mathbf{Q}$}
\STATE {$\mathbf{\Psi}$ =  DAVS-$Q$ ($\mathbf{D}, Q, \mathcal{G}$)}
\ENDFOR
\RETURN{$\mathbf{\Psi}$}
\end{algorithmic}
\end{algorithm}

\begin{algorithm}[t]
\caption{Candidate adjustment sets generation ($Candidate.gen$)}
\label{pse:Congen}
\noindent {\textbf{Input}}: $\mathbf{C}_{k-1}$; $k$.\\
\noindent {\textbf{Output}}: $\mathbf{C}_k$.
\begin{algorithmic}[1]
\STATE {$\mathbf{C}_k = \emptyset$;}
\STATE {Hash the first $(k-2)$ variables of all sets in $\mathbf{C}_{k-1}$}
\FOR {each pair $C_i, C_j\subseteq {\mathbf{C}_{k-1}}$ with the same first $(k-2)$ variables.}
\STATE{$C_s = \cup({C_i}, {C_j})$ }
\STATE {/*Generating all subsets of $C_s$ of size $k-1$.*/ \\ $Csubsets=subset(C_s, k-1)$ }
\IF {every element of $Csubsets$ is in $\mathbf{C}_{k-1}$}
\STATE {Add $C_s$ to $\mathbf{C}_k$}
\ENDIF
\ENDFOR
\RETURN{$\mathbf{C}_k$}
\end{algorithmic}
\end{algorithm}

\textbf{Complexity Analysis}. The time complexity of DAVS (same for DAVS-$Q$) is largely determined by the time taken by the rFCI algorithm~\cite{colombo2012learning} for learning the PAG from data $\mathbf{D}$. In the worst case, the time complexity of rFCI is $\textbf{O}(2^{p}n)$ in which $p$ denotes the number of nodes and $n$ implies the number of samples~\cite{colombo2012learning,claassen2013learning}. The average time complexity of rFCI normally much lower than its worst time complexity. rFCI is feasible for large graph~\cite{colombo2012learning}. Line 8 DAVS-$Q$ (Algorithm~\ref{pseudocode01}) takes $\textbf{O}(p)$ since it read off the graph.
Lines 11-22 in DAVS-$Q$ (Algorithm~\ref{pseudocode01}) in the worst case, the time complexity is $\textbf{O}(2^{l})$ where $l$ denotes the size of $pds(W, Y, \mathcal{G})\setminus (\{Q\}\cup Forb(W, Y, \mathcal{G}))$, which is lower than the complexity of the rFCI algorithm~\cite{colombo2012learning}. The time complexity of Algorithm~\ref{pse:Congen} is $\textbf{O}(l)$. Hence, in the worst case, the time complexity of DAVS-$Q$ is $\textbf{O}(2^{p}n + p + 2^{l} + l)$ = $\textbf{O}(2^{p}n)$ where the first term is due to rFCI. The overall time complexity of DAVS-$Q$ is the same as that of rFCI. For DAVS, the complexity of DAVS relies on $|\mathbf{Q}|$. Therefore, the overall complexity of DAVS is $\textbf{O}(|\mathbf{Q}|2^{p}n)$. With DAVS, the size of $\mathbf{Q}$ is normally small. Therefore, the complexity of DAVS-$Q$ and DAVS are determined by the rFCI algorithm.

\section{Experiments}
\label{Sec:Exp}
The overall goal of the experiments is to evaluate the performance of the proposed DAVS-$Q$ and DAVS algorithms. We categorise the comparison methods into two groups. The first group includes IDA, LV-IDA, and GAC, which do not have the pretreatment assumption and work on all datasets. The second group includes the most known causal effect estimators, which have the pretreatment assumption. In the second group, apart from those discussed in the Introduction, we also include three statistical causal effect estimation methods which are rooted in the potential outcome model~\cite{rubin1996matching,rubin2005causal}, including Propensity score matching (PSM)~\cite{rosenbaum1983central}, covariate balancing propensity score (CBPS)~\cite{imai2014covariate} and causal forest (CF)~\cite{wager2018estimation}. All three methods assume a known covariate set.

More details about all the comparison methods will be provided in the Related Work section.

\textbf{Implementation}. For IDA, we use the implementation in the $\mathbb{R}$ package \emph{pcalg}~\cite{kalisch2012causal} and for LV-IDA we use its implementation at the \emph{Github} site\footnote{\url{https://github.com/dmalinsk/lv-ida}}. The GAC is conducted by using the $\mathbb{R}$ packages \emph{pcalg} and \emph{Dagitty}~\cite{textor2011dagitty}. DICE is implemented using the functions \emph{pc.select}, \emph{GaussianCItest} and \emph{binCItest} in the $\mathbb{R}$ package \emph{pcalg}. The implementation of EHS~\cite{entner2013data} is from the authors\footnote{\url{https://sites.google.com/site/dorisentner/publications/CovariateSelection}}, and the implementation of CFRNet is also by its authors\footnote{\url{https://github.com/clinicalml/cfrnet}}. For \emph{CovSel}~\cite{de2011covariate,haggstrom2018data}, specifically for the five covariate selection methods in \emph{CovSel}, namely All causes of $W$ as the adjustment variables ($\hat{\mathbf{X}}_{\rightarrow W}$), All causes of $Y$ ($\hat{\mathbf{X}}_{\rightarrow Y}$), All causes of both $W$ and $Y$ ($\hat{\mathbf{X}}_{\rightarrow W,Y}$), Causes of $W$ excluding those independent of $Y$ ($\hat{\mathbf{Z}}_{\rightarrow W}$) and causes of $Y$ excluding those independent of $W$ ($\hat{\mathbf{Z}}_{\rightarrow Y}$), we use the implementations in the $\mathbb{R}$ package \emph{CovSelHigh} by H\'{a}ggstr\"{o}m~\cite{haggstrom2018data}. The classic PSM~\cite{rosenbaum1983central} is implemented by the function \emph{glm} in the $\mathbb{R}$ package \emph{stats} and the function \emph{Matching} in $\mathbb{R}$ package \emph{Matching}~\cite{ho2007matching}. CBPS implementation is from the $\mathbb{R}$ package \emph{CBPS}. CF implementation is from the $\mathbb{R}$ package \emph{grf}~\cite{wager2018estimation}.

\textbf{Parameter settings}. The same parameter settings are used for DAVS-$Q$, DAVS, and LV-IDA, and the significance level ($\alpha$) is set to 0.05. Since IDA, LV-IDA, EHS, and DICE each return a multiset of causal effects, the average of the causal effects is considered as the most probable estimation of the causal effects. For EHS, we constraint the size of the conditional set to 6. Otherwise, it cannot produce a result within two hours for the real-world datasets used.

\textbf{Evaluation metrics}. For a dataset with the ground truth causal effect, we evaluate the performance of algorithms using the estimated bias (\%), which is the relative error w.r.t. the ground truth. For a dataset without the ground truth of causal effect, we evaluate the algorithms using domain experts' suggested range or the consistency with the results among algorithms.

\subsection{Experiments with synthetic data}
\label{subsec:synthetic}
We use 5 benchmark Bayesian networks (BNs) from the Bayesian Network Repository\footnote{\url{http://www.bnlearn.com/bnrepository/}}: CHILD, INSURANCE, MILDEW, ALARM, and BARLEY to generate synthetic datasets. The numbers of nodes and arcs of the BNs are summarized in Table ~\ref{tab:detailedbn}. For each BN, we choose a variable with multiple incoming edges as the outcome variable $Y$, and select one of $Y$'s parents as the treatment variable $W$. We generate 5 synthetic datasets from the 5 BNs with 10,000 samples each by using the $\mathbb{R}$ package \emph{bnlearn}~\cite{scutari2009learning}. Then we hide 5\% variables which are on back-door paths from $W$ to $Y$ from each synthetic dataset. For the convenience of reproducing the experiment results in this paper, we provide the information of $W$, $Y$, the hidden variables and the COSO variable $Q$ for DAVS-$Q$ on the five synthetic datasets in Table~\ref{tab:variales}.

\begin{table}[htbp]
\scriptsize
\caption{A SUMMARY OF THE five standard benchmark BNs.}
\label{tab:detailedbn}
\centering
\begin{tabular}{|c|c|c|c|c|c|}
\hline
  & CHILD & INSURANCE  & MILDEW & ALARM & BARLEY\\ \hline
Nodes  & 20 & 27& 35 & 37 & 48 \\ \hline
Arcs  & 25 & 52  & 46 & 46 & 84\\  \hline
\end{tabular}
\end{table}

\begin{table}[htbp]
\tiny
\caption{The variables $W$, $Y$ and the hidden variables on the five synthetic datasets.}
\label{tab:variales}
\centering
\begin{tabular}{|p{0.055\textwidth}|p{0.055\textwidth}|p{0.05\textwidth}|p{0.11\textwidth}|p{0.06\textwidth}|}
\hline
BNs & $W$ & $Y$ & Hidden variables  & $Q$\\  \hline
CHILD  & LowerBodyO2 &  CardiacMixing &  LungParench & Disease  \\ \hline
INSURANCE & RiskAversion & DrivHist &    Age and MakeModel & GoodStudent \\ \hline
MILDEW & meldug$\_$3   & meldug$\_$4 & temp$\_$2 and temp$\_$3     & middel$\_$2 \\ \hline
ALARM & VALV & CCHL & PMB and SHNT & VLNG   \\ \hline
BARLEY & dgv1059 & protein & saa1id, dgv5980 and ng1ilg & dg25    \\ \hline
\end{tabular}
\end{table}

As the pretreatment variable assumption does not hold in these datasets, only methods in the first group without assuming pretreatment variables, i.e. IDA, LV-IDA, and GAC, are used in the comparisons on the five datasets. The causal effects calculated using Eq.(\ref{eq:eq01}), and the adjustment set identified using the back-door criterion~\cite{pearl2009causality} applied to the structures (DAGs) of the BNs are the ground truth causal effects.

\begin{table*}[ht]
\small
%\footnotesize
\centering
\caption{Estimating $ACE$ on synthetic datasets. The methods in Group 2 are not compared since they need pretreatment assumption to be hold or known covariates.}
\begin{tabular}{|l|cc||cc||cc||cc||cc|}
  \hline
\multirow{2}{*}{Methods}&
    \multicolumn{2}{|c||}{CHILD}&\multicolumn{2}{c||}{INSURANCE}&\multicolumn{2}{c||}{MILDEW}
    &\multicolumn{2}{c||}{ALARM} &\multicolumn{2}{c|}{BARLEY}\cr
    \cline{2-11}
& ACE &Bias(\%) & $ACE$ & Bias(\%)&  $ACE$ & Bias(\%)& $ACE$ & Bias(\%)& $ACE$ & Bias(\%)\\ \hline
IDA & -0.435 & 75.29\%  & -0.498 & 75.96\% &  0.332  & 75.98\% & 0.347 & 72.33\% &  -0.277 & 73.81\% \\
LV-IDA & -0.435 &  75.29\% & -0.481  & 76.75\%  & 0.344  &75.16\%& 0.333   &  73.42\%
& -0.260 & 75.42\% \\
   GAC &  -0.435 & 75.15\%  & -0.481 & 76.75\% &  0.343  & 75.17\% & 0.285 & 77.23\% &  -0.246 & 76.76\% \\
DAVS-$Q$ & \textbf{-1.793} & \textbf{1.81}\% & \textbf{-2.097} & \textbf{2.21}\% & \textbf{1.391} & \textbf{3.11}\%& \textbf{1.370} &  \textbf{9.27}\% & -1.027 & 3.02\% \\
DAVS & \textbf{-1.793}  & \textbf{1.81}\%& -2.179 & 5.25\% & 1.435 & 3.78\% &  0.836  & 33.31\% &  \textbf{-1.053}& \textbf{0.55}\% \\  \hline
\end{tabular}
\label{tab:BNs}
\end{table*}

From the results in Table~\ref{tab:BNs}, DAVS-$Q$ and DAVS have significantly smaller biases than IDA, LV-IDA, and GAC.
%The synthetic experimental results have demonstrated that our developed two algorithms not only can provide unique causal effect estimation, and but also more accurate estimations of the causal effects.
Because of orientation errors, the PAGs learned from the put for LV-IDA and GAC are inconsistent with the underlying MAGs for data generation, and hence both read off the invalid adjustment set from the learned PAG and do not perform well. IDA misses hidden generalised back-door paths and hence produces biased estimation. In contrast, DAVS and DAVS-$Q$ are less affected by orientation errors and have achieved very good results. DAVS-$Q$ and DAVS perform the conditional independence test $Q\Vbar Y\mid \mathbf{Z}\cup \{W\}$ to identify the proper adjustment set $\mathbf{Z}$ directly from data. This is the main reason for DAVS-$Q$ and DAVS to outperform LV-IDA and GAC.

%IDA handles data without hidden variables, so its poor performance is as expected.  Moreover, the adjustment set $\mathbf{Z}$ identified by LV-IDA and GAC heavily rely on the correctness of the learned PAG since the adjustment set $\mathbf{Z}$ are determined from the manipulated PAG but not identified directly from data. Different from LV-IDA and GAC, DAVS-$Q$ and DAVS perform the conditional independence test $Q\Vbar Y\mid \mathbf{Z}\cup \{W\}$ to identify the proper adjustment set $\mathbf{Z}$ directly from data. This is the main reason that the performance of DAVS-$Q$ and DAVS outperform the performance of LV-IDA and GAC.

\subsection{Experiments on two semi-synthetic real-world datasets}
\label{subsec:realworld}
We evaluate the performance of DAVS on two semi-synthetic real-world datasets, IHDP~\cite{hill2011bayesian} and Twins~\cite{almond2005costs,louizos2017causal}. Because the underlying causal graphs of the datasets are unknown, and no domain knowledge is available for nominating the  COSO variable $Q$ needs to be given by domain experts or knowledge, so DAVS-$Q$ will not be evaluated on these datasets. Both datasets contain pretreatment variables only in addition to the treatment and outcome variables, and hence we can include all methods in Table~\ref{tab_methodClassificaiton} (excluding CE-SAT since its high complexity) in the comparison.

\subsubsection{IHDP}
\label{subsubsec:IHDP}
IHDP is the Infant Health and Development Program (IHDP) dataset collected from a randomized controlled experiment that investigated the effect of home visits by specialists on future cognitive test scores~\cite{hill2011bayesian}. There are 24 pretreatment variables and 747 infants, including 139 treated (having home visits by specialists) and 608 controls. The simulated outcome variables were generated by using the $\mathbb{R}$ package \emph{npci}~\cite{Dorie2016nonparmetrics} and the factual outcomes $Y$ was generated with the \emph{noiseless} outcome according to the same procedures suggested in Hill~\cite{hill2011bayesian}. The generated procedures allow obtaining the ground truth $ACE(W, Y)= 4.36$.

\begin{table}
\small
\centering
\caption{Estimating $ACE$ on IHDP and Twins. The lowest biases in each group of methods are highlighted. DAVS-$Q$ is not applicable since we do not know $Q$. }
\begin{tabular}{|l|cc||cc|}
  \hline
\multirow{2}{*}{Methods}&
    \multicolumn{2}{|c||}{IHDP}&\multicolumn{2}{c|}{Twins}\cr\cline{2-5}  &  $ACE$  & Bias(\%) &      $ACE$    &  Bias(\%)\\ \hline
 IDA                                 &  3.01  &  31.03\%  & -0.023    & 7.29\%   \\
 LV-IDA                               &  3.81 &  12.62\%  &   -0.021  & 14.71\%      \\
 GAC & 3.81 &  12.62\% &   -0.021  & 14.71\%\\
 DAVS                           & \textbf{4.00} & \textbf{8.25}\% & \textbf{-0.024}  & \textbf{3.46}\%  \\ \hline
 CFRNet                   &  4.15 &  4.82\%  & -0.020 & 20.25\%  \\
 DICE & \textbf{4.48} & \textbf{2.50}\%   & -0.038 & 51.38\% \\
 EHS                                 &  4.05 &  7.07\%   &   -0.033  & 24.27\% \\
$\hat{\mathbf{X}}_{\rightarrow W}$   &  3.59 &  17.66\%  &  -0.021  &14.70\%     \\
$\hat{\mathbf{Z}}_{\rightarrow W}$            &  3.69 &  15.37\%  &   -0.019  & 21.84\%    \\
$\hat{\mathbf{X}}_{\rightarrow Y}$   &  3.71 &  14.91\%  &   -0.013  & 46.72\%    \\
$\hat{\mathbf{Z}}_{\rightarrow Y}$            &  3.70 &  15.14\%  &   -0.012  & 50.65\%    \\
$\hat{\mathbf{X}}_{\rightarrow W, Y}$&  4.13 &  5.27\%   &   -0.017  & 31.63\%     \\
PSM                                  & 3.94  &  9.63\%   &  -0.016   & 35.83\% \\
CBPS &  4.07  & 6.54\%  & \textbf{-0.027} & \textbf{6.90}\%  \\
CF &3.50 &  19.64\% & -0.017 & 30.14\%\\
\hline
\end{tabular}
\label{tab:IHDPtwins}
\end{table}

\subsubsection{Twins}
\label{subsubsec:Twins}
Twins is a benchmark dataset used in causal inference literature. It is about twin births and deaths in the USA from 1989 -1991~\cite{almond2005costs}. We only choose same-sex twins with weights less than 2000g from the original dataset. Each twin-pair contains 40 pretreatment variables related to the parents, the pregnancy, and the birth of the twins~\cite{louizos2017causal}. We eliminate all records with missing values, so 4821 twin-pairs remain. For each pair, we observe both the treated ($W$=1, the heavier twin) and control ($W$=0, the lighter twin). The mortality after one year is the true outcome for each twin such that the ground truth $ACE(W, Y)$ is -0.02489. For simulating an observational study, we follow~\cite{louizos2017causal} to randomly hide one of the two twins, by using the setting: $W_i|x_i\sim Bern(sigmoid(\beta^{T}x+\varepsilon))$, where $\beta^{T}\sim\mathcal{U}((-0.1, 0.1)^{40\times1})$ and $\varepsilon\sim\mathcal{N}(0,0.1)$.

From the results in Table~\ref{tab:IHDPtwins}, we see that DAVS outperforms IDA, LV-IDA, and GAC in the first group. The second group of methods needs a stronger requirement (pretreatment variables only or known covariates), which is satisfied by the two datasets. With the Twins dataset, DAVS outperforms all other methods. With the IHDP dataset, DAVS outperforms most comparison methods. Considering that IHDP is a small dataset and DAVS needs many conditional independence tests, this performance is very impressive. When the size of a dataset increases (from IHDP to Twins), the performance of the methods in the first group improves significantly competitively with the performance of the methods of the second group.

\subsection{Experiments on two real-world datasets}
\label{subsec:tworealworld}
In this section, we further evaluate the performance of DAVS on two real-world datasets, i.e. Cattaneo2~\cite{cattaneo2010efficient} and RHC~\cite{connors1996effectiveness}. Since both datasets have no ground truth causal effects available, we judge the performance of DAVS according to whether it achieves a performance consistent with other methods and empirically estimated interval suggested in the literature.

\subsubsection{Cattaneo2}
\label{subsubsec:cat}
The Cattaneo2 dataset~\cite{cattaneo2010efficient} is frequently utilised to investigate the $ACE$ of mother's smoking status during pregnancy ($W$) on a baby's birth weight ($Y$ in grams)\footnote{\url{http://www.stata-press.com/data/r13/cattaneo2.dta}}.
The dataset contains birth weights of 4642 singleton births in Pennsylvania, USA~\cite{almond2005costs,cattaneo2010efficient}, with 864 smoking mothers ($W$=1) and 3778 non-smoking mothers ($W$=0). It contains 20 covariates and names a few of them, including mother's age, mother's marital status, mother's race, mother's education, father's education. The authors of~\cite{almond2005costs} have drawn a conclusion that there is a strong negative effect of maternal smoking on the weights of babies, about $200g$ to $250g$ lighter for a baby with a mother no-smoking during pregnancy (the empirically estimated interval).

The experimental results of all methods are listed in Table~\ref{tab:Catt_rhc} and visualized in Fig.~\ref{fig:cattaneo2}. We see that the estimated range of maternal smoking on babies' birth weight is $-275.25g$ to $-120.88g$. Furthermore, the estimated effects by IDA, DAVS, DICE, $\hat{\mathbf{X}}_{\rightarrow W}$ and $\hat{\mathbf{Z}}_{\rightarrow Y}$ fall in the empirically estimated interval ($-250g, -200g$) as shown in Fig.~\ref{fig:cattaneo2}, while the estimations by the other methods do not fall in the empirically estimated interval. DAVS has better performance than most of the state-of-the-art methods, and its estimate is consistent with the empirically estimated effect~\cite{almond2005costs}.
%Overall, the experimental results provide strong evidence that smoking leads to lower birth weight.
%, and it can be used for real-world applications.

\begin{figure}
\includegraphics[width=8.9cm,height=4.4cm]{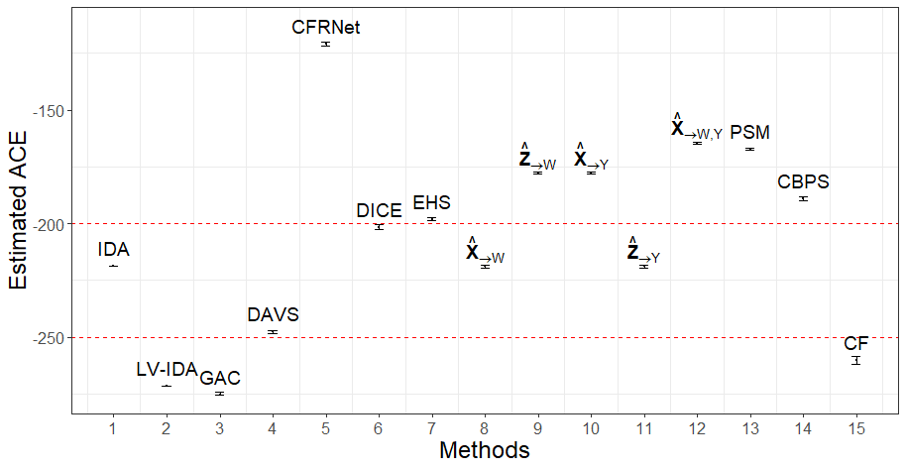}
\caption{The estimated causal effects by different methods in 95\% confident intervals for dataset Cattaneo2. The two dotted lines denote the empirical estimated interval $(-250g, -200g)$ suggested in the literature.}
\label{fig:cattaneo2}
\end{figure}

\subsubsection{Right Heart Catheterization}
\label{subsubsec:rhc}
The dataset Right Heart Catheterization (RHC) is from an observational investigation, which studies a diagnostic procedure for the management of critically ill patients~\cite{connors1996effectiveness}. The RHC dataset can be found in the $\mathbb{R}$ package \emph{Hmisc}\footnote{\url{https://CRAN.R-project.org/package=Hmisc}}. The dataset consists of the information of hospitalized adult patients from five medical centers in the USA. These hospitalized adult patients participated in the Study to Understand Prognoses and Preferences for Outcomes and Risks of Treatments (SUPPORT). Whether or not a patient received an RHC within 24 hours of admission is the treatment $W$. Whether a patient died at any time up to 180 days after admission is the outcome $Y$. The original data contains 5735 samples with 73 covariates. We pre-process the original data, as suggested by Loh et al.~\cite{loh2020confounder}, and obtain the dataset which contains 2707 samples with 72 covariates.

\begin{table}
%\tiny
\small
%\footnotesize
%\centering
\caption{Estimating $ACE$ on the Cattaneo2 and RHC datasets. In Cattaneo2 dataset, five estimates highlighted are in the empirical estimated interval $(-250g, -200g)$ suggested in the literature. In RHC dataset, the overall range of all estimators is small and hence the estimated causal effects of all methods are considered consistent.  EHS did not return results in two hours in dataset RHC. DAVS-$Q$ is not applicable since we do not know $Q$.}
\hspace*{-0.2 cm}
\begin{tabular}{|l|cc||cc|}
  \hline
\multirow{2}{*}{Methods}&
    \multicolumn{2}{|c||}{Cattaneo2}&\multicolumn{2}{c|}{RHC}\cr\cline{2-5}  &  $ACE$  & 95\% C.I. &      $ACE$    &  95\% C.I.\\ \hline
 IDA  &  \textbf{-218.8} & $\diagup$ &0.0666 & $\diagup$  \\
 LV-IDA & -271.4 &  (-271.5, -271.3)  & 0.0400  & (0.0397, 0.0404)   \\
 GAC & -275.3&  (-275.5, -274.3) &  0.0400  &(0.0397, 0.0404)\\
 DAVS  & \textbf{-247.6} & (-248.3, -246.9) & 0.0395& (0.0394, 0.0396) \\ \hline
 CFRNet   &-120.9 & (-121.6, -120.1)&0.0129 &(0.0103, 0.0155) \\
 DICE & \textbf{-201.6} &(-202.7, -200.6)& 0.0290 & (0.0282, 0.0299) \\
 EHS  & -198.1&(-198.6, -197.5)&$\diagup$&$\diagup$ \\
$\hat{\mathbf{X}}_{\rightarrow W}$& \textbf{-218.8} &(-219.5, -218.2)&0.0407&(0.0399, 0.0415) \\
$\hat{\mathbf{Z}}_{\rightarrow W}$&-177.6&(-178.1, -177.1)&0.0156&(0.0150, 0.0163) \\
$\hat{\mathbf{X}}_{\rightarrow Y}$   &  -177.6 & (-178.1, -177.1)    & 0.0280 & (0.0273, 0.0286)   \\
$\hat{\mathbf{Z}}_{\rightarrow Y}$  &  \textbf{-218.8} & (-219.5, -218.2)  &  0.0191  & (0.0183, 0.0198)\\
$\hat{\mathbf{X}}_{\rightarrow W, Y}$& -164.7 & (-165.2, -164.2)& 0.0522  & (0.0515, 0.0529)    \\
PSM                         &-167.2  &  (-167.7, -166.7)   &  0.0614  & (0.0606, 0.0621) \\
CBPS &  -189.2  & (-190.1, -188.2) & 0.0209 & (0.0193, 0.0225)  \\
CF    & -260.1 & (-261.8, -258.4) & 0.0256 & (0.0243, 0.0269)\\
\hline
\end{tabular}
\label{tab:Catt_rhc}
\end{table}

The experimental results on RHC are listed in Table~\ref{tab:Catt_rhc} and visualized in Fig.~\ref{fig:RHC} (excluding EHS since it can not return results in two hours). The estimated $ACE$s by different methods range from 0.0129 to 0.0666, and the difference is very small, and thus the estimated causal effects of all methods are consistent. \\
%Overall, the estimated causal effects of all methods (including our developed DAVS method) show that applying RHC results in higher mortality with 180 days than the patient who does not apply RHC.

\begin{figure}
\includegraphics[width=8.9cm,height=4.4cm]{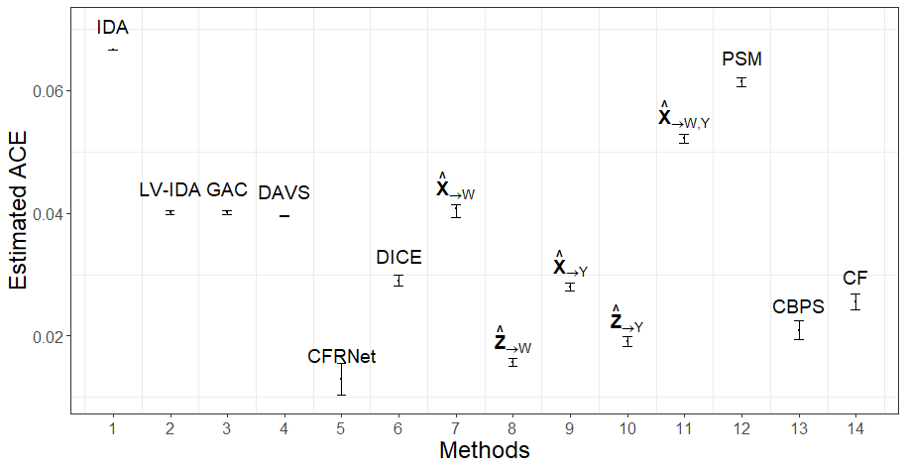}
\caption{The estimated causal effects by different methods in  95\% confident intervals for the RHC dataset. The maximal difference of all estimates is within 0.054 which is very small. The estimates are consistent. }
\label{fig:RHC}
\end{figure}

In sum, DAVS is consistently a top performer in all datasets tested and is the only method that achieves this. DAVS obtains the estimates which are approximate to the true or empirically estimated causal effects on both synthetic and real-world datasets. DAVS can be used in a range of datasets, with or without missing values, whether or not the pretreatment variable assumption is satisfied. %compensative to the-state-of-art comparison methods.

%The synthetic experiments have confirmed that the developed two algorithms can handle data well, which violates the assumptions of pretreatment variables and unconfoundedness. It means that the developed two algorithms have a wider range of applications than the estimated methods, which requires the assumptions of either pretreatment variables or unconfoundedness or both.

\subsection{Efficiency evaluation}
\label{Subsec:Efficiency}
We evaluate the time efficiency of DAVS in comparison to the methods in the first group. We do not compare the methods in the second group since they require the pretreatment assumption or known covariates, and thus they do not work with all the datasets.

All computations were conducted on a PC with \emph{2.6GHz IntelCore i7} and \emph{16GB} of memory. The runtime of the algorithms on all datasets (synthetic and real-world) is shown in Fig.~\ref{fig:runtime}. All methods have similar time efficiency. EHS~\cite{entner2013data} is related to our method since it also uses a COSO variable, but EHS is extremely time-consuming and can only be run on datasets with few variables since it exhaustively searches for the adjustment sets. In the RHC dataset, EHS did not return results in two hours while DAVS returns results in 11 mins. Also, EHS is applicable to more restrictive datasets than DAVS.

\begin{figure}
\includegraphics[width=8.9cm,height=4.4cm]{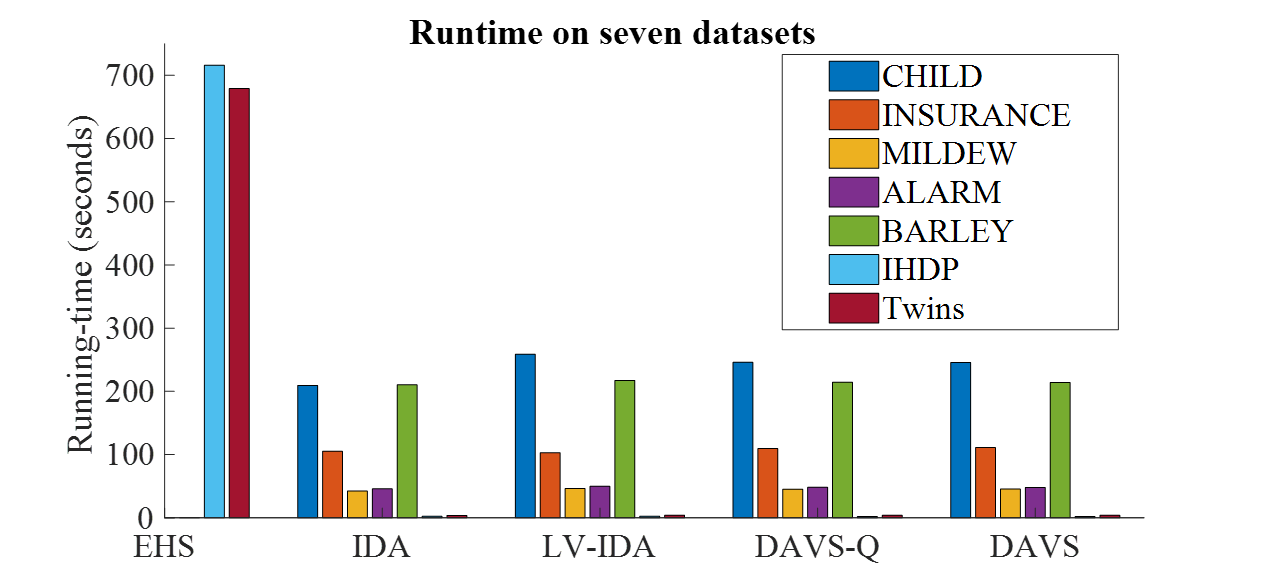}
\caption{Runtime in seconds of five methods on 9 datasets. Note that DAVS-$Q$ only works on five synthetic datasets. The time efficiency of all methods are similar.}
\label{fig:runtime}
\end{figure}

\subsection{Sensitive Analysis for different $Q$}
\label{Subsec:SenA}
To investigate the sensitivity of DAVS to the selection of different COSO variables $Q$, we analyze the estimated causal effects by DAVS on the five synthetic datasets described in Section~\ref{subsec:synthetic}. Since there exists only one single COSO variable in the causal structures of the ALARM and CHILD datasets, we exclude them from this analysis. We visualize the estimated causal effects under different COSO variables in Fig.~\ref{fig:COSO01} on the three remaining synthetic datasets, INSURANCE, MILDEW, and BARLEY. From Fig.~\ref{fig:COSO01}, we can see that the estimated causal effects are very close to the true causal effect under different COSO variables $Q$ on all of three synthetic datasets. In other words, the estimated causal effects are not sensitive to the COSO variable selection since in all cases, the effects estimated all approximate the true causal effect. %It further confirms the soundness of Theorem~\ref{theo:theo03}. %To refine the covariate set,  \emph{CovSel}~\cite{de2011covariate,haggstrom2018data} reduce the covariate set using various criteria.

\begin{figure*}
  \centering
\begin{minipage}[t]{0.3\linewidth}
\centering
\includegraphics[width=5.8cm,height=4cm]{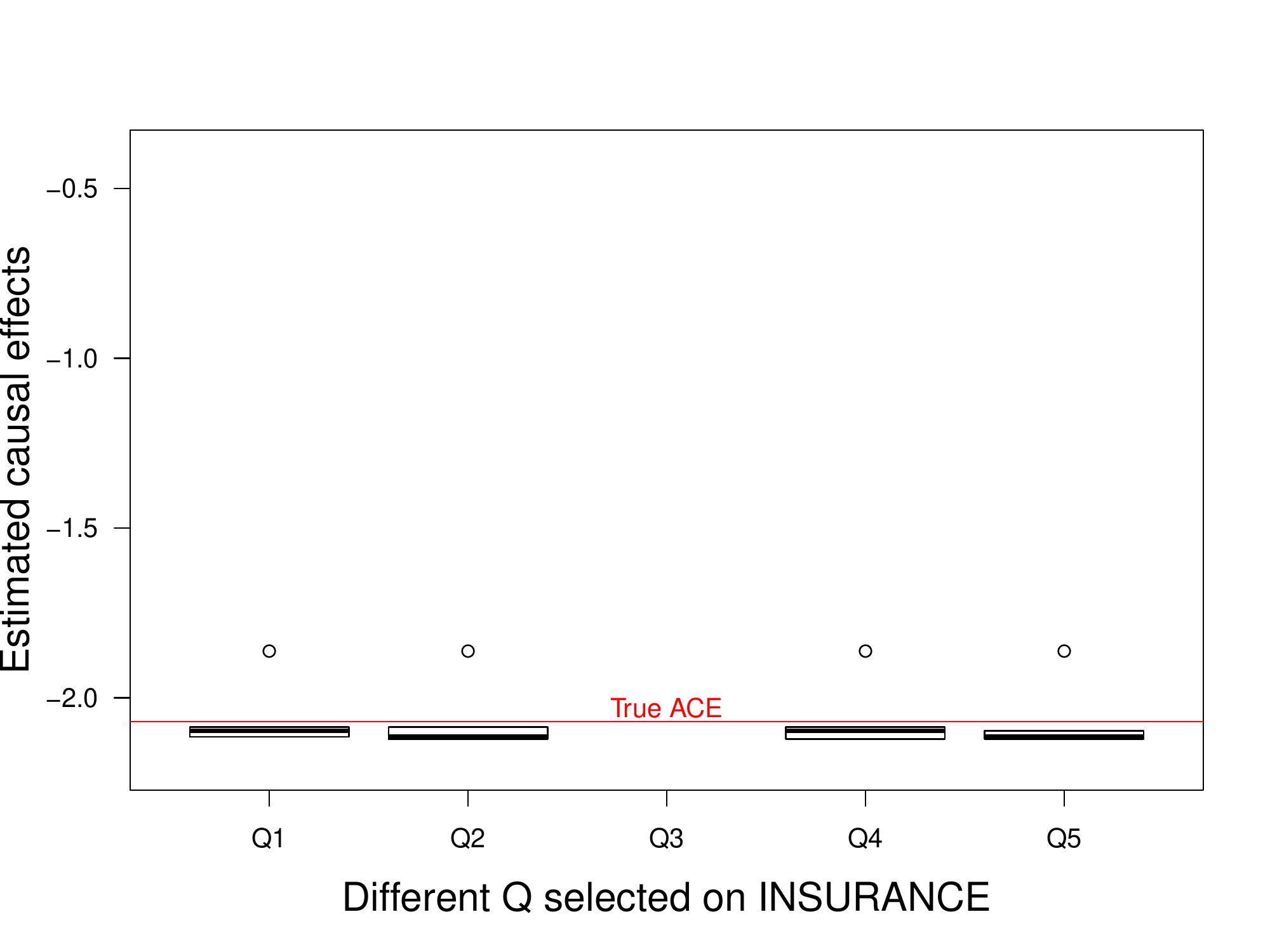}
\end{minipage}
\hfill
\begin{minipage}[t]{0.3\linewidth}
\centering
\includegraphics[width=5.8cm,height=4cm]{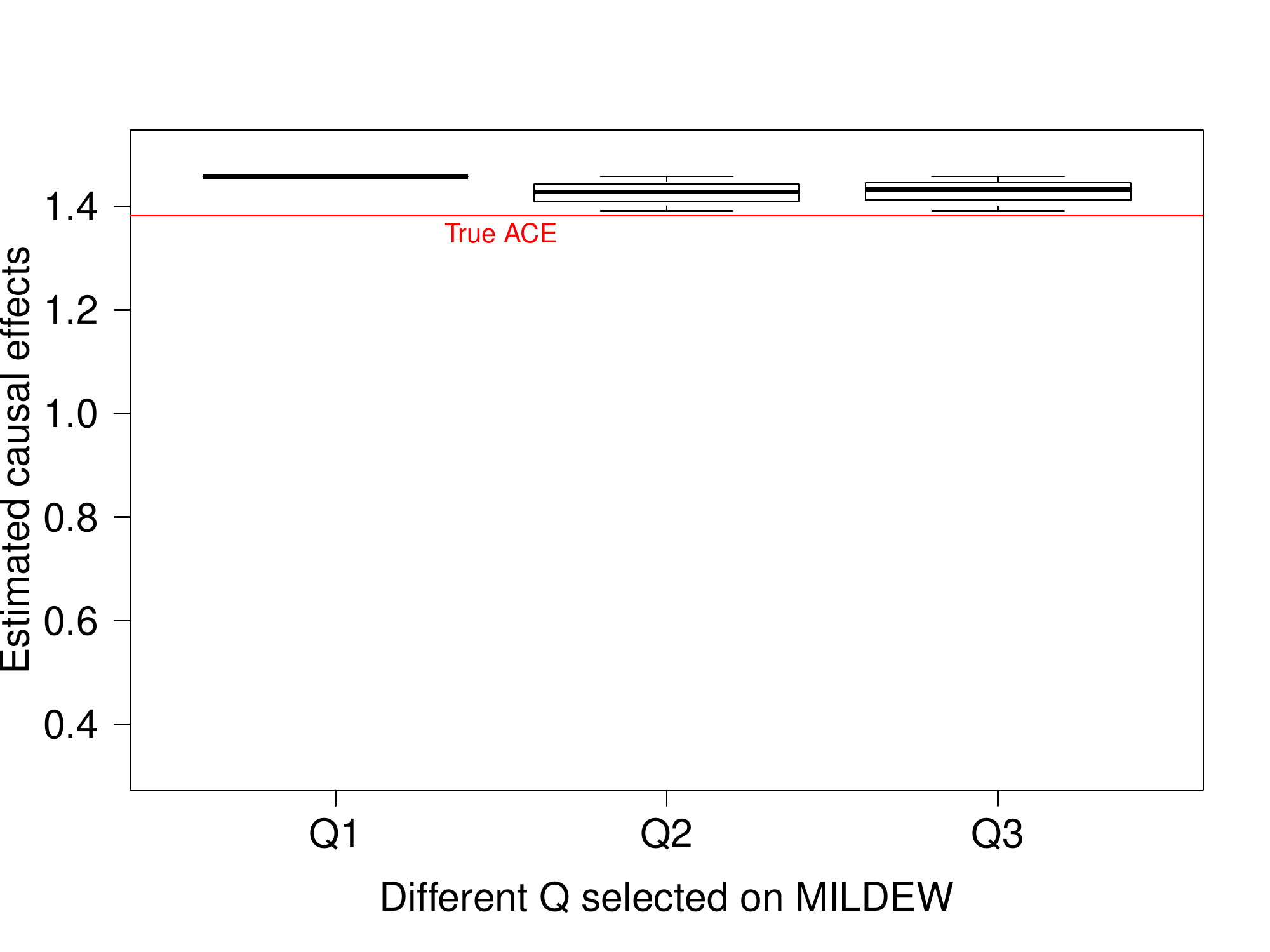}
\end{minipage}
\hfill
\begin{minipage}[t]{0.3\linewidth}
\centering
\includegraphics[width=5.8cm,height=4cm]{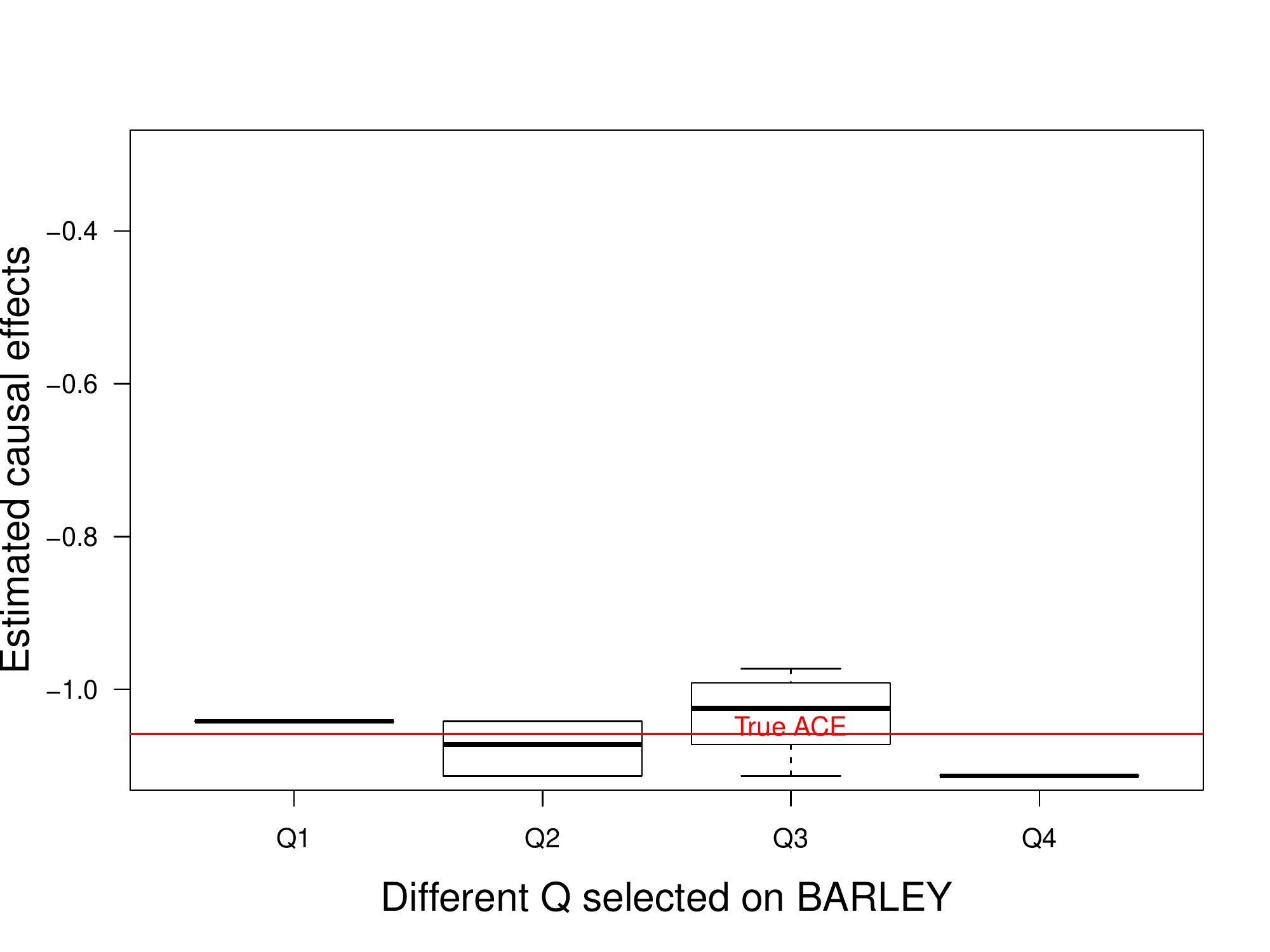}
\end{minipage}
\caption{The boxplots of the estimated causal effects under different COSO $Q$ on three synthetic datasets, i.e. INSURANCE, MILDEW and BARLEY which have more than one COSO variable. The range is taken from the smallest causal effect and the largest causal effect in Table~\ref{tab:BNs}. For each different COSO $Q$, the estimated causal effects are consistent and close the true causal effect on all of three datasets.} %The range is the same as causal effects estimated in Table~\ref{tab:BNs}. With different COSO $Q$s, the estimated causal effects are consistent and close the true causal effect on all of three datasets.}
  \label{fig:COSO01}
\end{figure*}

\section{Related work}
\label{Sec:related}
Causal effect estimation from observational data has been an active research topic in statistics and artificial intelligence, and is largely based on two frameworks: potential outcome framework~\cite{rubin1974estimating,imbens2015causal} and graphic causal modelling~\cite{pearl2009causality,perkovic2017complete,Spirtes2000Causation}. There are also many research works on estimating two different types of causal effect, the average causal effect estimation in a whole population and conditional causal effect estimation in a subpopulation. This work belongs to the first type and we will briefly outline some works on conditional causal effect estimation at the end of this section.

For causal effect estimation, various adjustment criteria~\cite{pearl2009causality,de2011covariate,shpitser2012validity,perkovic2017complete,van2019separators} (or in general using $do$ calculus operation~\cite{pearl2009causality}) are commonly employed to approach the problem in graphical causal modelling. The methods in this category have been outlined in the Introduction, and hence we do not discuss them here. For the methods under the potential outcome framework, the covariate set is normally assumed that satisfies the unconfoundedness and a matching method based on covariate balance is used for causal effect estimation~\cite{li2017matching,rubin1974estimating,rubin2007design}. We do not intend to review the methods here and refer interested readers to~\cite{imbens2015causal,witte2019covariate}. Instead, we just briefly discuss representative methods based on potential outcome framework used in this paper. Matching is a practical way to balance the distributions of covariates in the treatment and control groups to achieve unbiased causal effect estimation. When the size of a covariate set is large, exact matching or the nearest neighbour matching is impractical. The propensity score collapses the covariate set to one dimension and is a sound way for balancing covariate distribution in treatment and control groups. Propensity Score Matching (PSM)~\cite{rosenbaum1983central} is the classic method for average causal effect estimation and is widely used in various applications~\cite{imbens2015causal,austin2019propensity}. Covariate balancing propensity score (CBPS)~\cite{imai2014covariate} is a major improvement of propensity score matching and has been shown good performance. Causal forest (CF)~\cite{wager2018estimation} uses the random forest to improve causal effect estimation based on matching. A covariate set may include some variables that are not necessary to be adjusted in causal effect estimation. The process of selecting a refined covariate set is called covariate set selection~\cite{de2011covariate,shpitser2012validity,entner2013data}. The recently causal graph has been used for covariate set selection, and \emph{CovSel}~\cite{haggstrom2018data} summarises the most widely used covariate selection criteria. Five criteria have been used in this paper.

The conditional causal effect estimation aims to remove the heterogeneous by identifying the subpopulations with different causal effects. For estimating conditional causal effect, many machine learning algorithms have been adapted for the task~\cite{athey2016recursive,kunzel2019metalearners,kuang2017estimating,yao2018representation}. We mention some typical works here, and a comprehensive review can be found in~\cite{guo2020survey,zhang2020unified}. Athey, Wager et al.~\cite{athey2016recursive} have adapted a tree-based algorithm for estimation. Louizos, Shalit et al.~\cite{louizos2017causal} have proposed to utilise Variational Auto-encoder to estimate the unknown latent space of confounders and conditional causal effects simultaneously. Hassanpour and Greiner~\cite{hassanpour2019counterfactual} have proposed a context-aware important sampling re-weighing scheme to address the distributional shift due to selection bias in conditional causal effect. K{\"u}nzel, Sekhon et al.~\cite{kunzel2019metalearners} have developed a meta-learner, the X-learner, for conditional causal effect estimation. Our work differs from all these works in that we estimate the average causal effect. Because of different purposes, these methods have not been compared in the experiment.

\section{Conclusion}
\label{Sec:Con}
We have identified a practical problem setting, in which the causal effects of $W$ on $Y$ can be estimated uniquely and unbiasedly from observational data with hidden variables. In particular, under our problem setting, we have developed the two theorems to ensure that it is sufficient and necessary to perform condition independence tests on data to find a proper adjustment set for unbiased causal effect estimation. We have proposed two data-driven algorithms, DAVS-$Q$ and DAVS, to obtain unique and more accurate causal effect estimation from data with hidden variables. Experimental results on synthetic and real-world datasets demonstrate that our proposed algorithms have achieved better results for data with hidden variables compared to the state-of-the-art causal effect estimators. Compared to existing causal effect estimators, our proposed algorithms have relaxed the pretreatment variable and allowed hidden variables, and are able to provide a unique and unbiased estimation. This will significantly improve the practical applications of data-driven causal effect estimation.

\ifCLASSOPTIONcaptionsoff
  \newpage
\fi



\bibliographystyle{IEEEtran}
\bibliography{IEEEabrv}
\iffalse


\begin{IEEEbiography}{Michael Shell}
Biography text here.
\end{IEEEbiography}

\begin{IEEEbiographynophoto}{John Doe}
Biography text here.
\end{IEEEbiographynophoto}

%\newpage

\begin{IEEEbiographynophoto}{Jane Doe}
Biography text here.
\end{IEEEbiographynophoto}
\fi

%\vfill


\end{document}